\documentclass[12pt,a4paper]{article}
\usepackage[margin=2.2cm]{geometry}
\usepackage{authblk}

\usepackage[breaklinks]{hyperref}
\usepackage{booktabs}
\usepackage{mathptmx}

\SetSymbolFont{letters}{bold}{OML}{cmm}{b}{it}

\usepackage{amsthm}
\newtheorem{theorem}{Theorem}

\newtheorem{definition}[theorem]{Definition}
\newtheorem{lemma}[theorem]{Lemma}
\newtheorem{corollary}[theorem]{Corollary}

\newtheorem{example}[theorem]{Example}
\newenvironment{acknowledgement}{\smallskip\small\noindent\emph{Acknowledgement}}{}

\newcommand{\cmpqed}{}
\newcommand{\cmpbreak}{}
\newcommand{\arxivbreak}{\break}

\newif\ifnonCMP
\nonCMPtrue

\usepackage{tikz}
\usetikzlibrary{matrix,arrows,patterns}
\usepackage{soul,todonotes}
\usepackage{leftidx}
\usepackage{amsmath,amssymb,amsfonts}
\usepackage{slashed}
\newcommand{\bbLbrack}{[\kern-0.4em{[}\,}
\newcommand{\bbRbrack}{\,]\kern-0.4em{]}}

\usepackage{extarrows}
\usepackage{undertilde}

\usepackage{mathrsfs,latexsym}
\usepackage[mathscr]{eucal}
\usepackage{mathrsfs}

\newcommand{\Mbb}{{\boldsymbol{\EuScript{M}}}}
\newcommand{\Nbb}{{\boldsymbol{\EuScript{N}}}}
\newcommand{\Rbb}{{\boldsymbol{\EuScript{R}}}}
\newcommand{\Ubb}{{\boldsymbol{\EuScript{U}}}}

\newcommand{\II}{{\boldsymbol{1}}}

\newcommand{\CC}{{\mathbb C}}

\newcommand{\RR}{{\mathbb R}}

\newcommand{\NN}{{\mathbb N}}

\newcommand{\ZZ}{{\mathbb Z}}

\newcommand{\CoinX}[1]{C_0^\infty({#1})}

\newcommand{\DD}{{\mathscr D}}
\newcommand{\EE}{{\mathscr E}}
\newcommand{\HH}{{\mathscr H}}

\newcommand{\Gc}{{\mathcal G}}
\newcommand{\Hc}{{\mathcal H}}

\newcommand{\ogth}{{\mathfrak o}}
\newcommand{\tgth}{{\mathfrak t}}

\newcommand{\supp}{{\rm supp}\,}

\DeclareMathOperator{\Span}{span}

\DeclareMathOperator{\ad}{ad}

\renewcommand{\Im}{{\rm Im}\,}

\DeclareMathOperator{\Sym}{Sym}

\newcommand{\dvol}{d\textrm{vol}}

\newcommand{\WF}{{\rm WF}\,}

\newcommand{\Mb}{{\boldsymbol{M}}}
\newcommand{\Nb}{{\boldsymbol{N}}}

\newcommand{\Lc}{{\mathcal{L}}}
\newcommand{\Mc}{{\mathcal{M}}}
\newcommand{\Nc}{{\mathcal{N}}}
\newcommand{\Sc}{{\mathcal{S}}}

\newcommand{\Ct}{{\sf C}}

\newcommand{\BkGrnd}{{\sf BkGrnd}}

\newcommand{\Loc}{{\sf Loc}}
\newcommand{\CLoc}{{\sf CLoc}}
\newcommand{\FLoc}{{\sf FLoc}}
\newcommand{\SpinLoc}{{\sf SpinLoc}}

\newcommand{\Set}{{\sf Set}}

\newcommand{\Alg}{{\sf Alg}}

\newcommand{\Phys}{{\sf Phys}}

\newcommand{\Af}{{\mathscr A}}

\newcommand{\Bf}{{\mathscr B}}

\newcommand{\Cf}{{\mathscr C}}
\newcommand{\Df}{{\mathscr D}}

\newcommand{\Ff}{{\mathscr F}}

\newcommand{\Rf}{{\mathscr R}}
\newcommand{\Sf}{{\mathscr S}}
\newcommand{\Tf}{{\mathscr T}}
\newcommand{\Uf}{{\mathscr U}}
\newcommand{\Vf}{{\mathscr V}}

\newcommand{\FfL}{\Ff_{\sf L}}
\newcommand{\FfS}{\Ff_{\sf S}}

\newcommand{\Rc}{{\sf R}}

\newcommand{\Vc}{{\mathcal V}}
\newcommand{\Wf}{{\mathscr W}}

\newcommand{\id}{{\rm id}}

\newcommand{\nto}{\stackrel{.}{\to}}

\newcommand{\Fld}{{\rm Fld}}

\DeclareMathOperator{\Aut}{Aut}

\newcommand{\rce}{{\rm rce}}

\newcommand{\tc}{{\textrm{tc}}}

\makeatletter
\newcommand{\raisemath}[1]{\mathpalette{\raisem@th{#1}}}
\newcommand{\raisem@th}[3]{\raisebox{#1}{$#2#3$}}
\makeatother
\newcommand{\act}[2]{\leftidx{^{\raisemath{-1.5pt}{#1}}}{#2}{}}

\newcommand{\GL}{{\rm GL}}
\newcommand{\SL}{{\rm SL}}

\newcommand{\SO}{{\rm SO}}

\newcommand{\ett}{\widetilde{\eta}}
\newcommand{\xit}{\widetilde{\xi}}

\newcommand{\phit}{\widetilde{\phi}}

\begin{document}
\title{An analogue of the Coleman--Mandula theorem for quantum field theory in curved spacetimes} 

\author[1]{Christopher J. Fewster\thanks{\tt chris.fewster@york.ac.uk}}
\affil{Department of Mathematics,
	University of York, Heslington, York YO10 5DD, United Kingdom.}

\date{\today}
\maketitle
%
%
\maketitle
%
%
\begin{abstract} 
The Coleman--Mandula (CM) theorem states that the Poincar\'e and internal symmetries of a Minkowski spacetime quantum field theory
cannot combine nontrivially in an extended symmetry group. We establish an analogous result
for quantum field theory in curved spacetimes, assuming local covariance, the timeslice property, a local dynamical form of Lorentz invariance, and additivity.
Unlike the CM theorem, our result is valid in dimensions $n\ge 2$
and for free or interacting theories. It is formulated for theories defined on a category of all globally hyperbolic spacetimes equipped with a global coframe, on which the restricted Lorentz group acts, and makes use of a general analysis of  
symmetries induced by the action of a group $G$ on the category of spacetimes.
Such symmetries are shown to be canonically associated with a cohomology class
in the second degree nonabelian cohomology of $G$ with coefficients in the global
gauge group of the theory. Our main result proves that the cohomology class is trivial if
$G$ is the universal cover $\Sc$ of the restricted Lorentz group.  Among other consequences, it follows that the extended
symmetry group is a direct product of the global gauge group and $\Sc$, all fields transform in multiplets of $\Sc$, fields of different spin do not mix under the extended group, and the occurrence of noninteger spin is controlled by the centre of the global
gauge group. The general analysis is also applied to rigid scale covariance.
\end{abstract}
{\noindent\em Dedicated to the memory of Rudolf Haag}

\section{Introduction}
\label{sec:intro}
%


In the issue of Communications in Mathematical Physics 
dedicated to Rudolf Haag's 80th birthday, Brunetti, Fredenhagen
and Verch~\cite{BrFrVe03} introduced \emph{locally covariant quantum field theory}, 
a formulation of QFT in curved spacetimes that is a far-reaching generalization
of Haag's framework of local quantum physics~\cite{Haag} (also called algebraic QFT). 
Locally covariant QFTs are expressed as functors from a category of spacetimes 
$\BkGrnd$ to a category of physical systems $\Phys$. The morphisms of
$\BkGrnd$ correspond to embeddings of one spacetime as a subspacetime of another,
while the morphisms of $\Phys$ correspond to embeddings of one physical system
as a subsystem of another. A functor $\Af:\BkGrnd\to\Phys$ therefore associates
a physical system to every spacetime and also specifies how each spacetime
embedding gives an embedding of these physical systems. Thus, $\Af$ defines
the theory on all spacetimes and incorporates the principle of locality from the start. 
Locally covariant QFT has proved to be a fruitful framework for the general
analysis of QFT in curved spacetime and has allowed various structural
results or properties of flat spacetime QFT to be transferred to  
curved spacetimes (see~\cite{FewVerch_aqftincst:2015} for a review). Examples include the spin-statistics
connection~\cite{Verch01}, the analysis of superselection sectors~\cite{Br&Ru05}, Reeh--Schlieder and split properties~\cite{Sanders_ReehSchlieder,Few_split:2015}, punctured Haag duality~\cite{Ruzzi_punc:2005}, and modular nuclearity~\cite{LechSan:2016}; one can also
discuss the question of whether 
a theory represents the same physics
in all spacetimes~\cite{FewVer:dynloc_theory}.   These ideas
also play a central role in constructions of perturbative QFT
in curved spacetimes~\cite{Ho&Wa01,Ho&Wa02,Rejzner_book}.  

The aim of this paper is to formulate and prove an analogue of the Coleman--Mandula (CM) theorem~\cite{ColeMand:1967} for locally covariant QFT on general parallelizable
globally hyperbolic spacetimes of dimension $n\ge 2$. 
The CM theorem originated as part of an intensive
effort in the 1960's to understand whether the internal and
Poincar\'e symmetries of a QFT in Minkowski space could be combined (`mixed') in a larger
symmetry group other than as a direct product.
These investigations led to a series
of no-go theorems of increasing
scope based on group theoretic grounds~\cite{Michel:1964,ORaif:1965,Jost:1966} or,
as with the CM result itself (and its generalizations to dimensions $n>4$~\cite{PelcHor:1997}), on dynamical considerations centred on the $S$-matrix. Later, supersymmetry offered a potential loophole to
these results, because fermionic charges interchange bosonic and fermionic fields and therefore also change spin. One of Haag's most highly cited papers was his joint work with {\L}opusza\'{n}ski and Sohnius~\cite{HaagLopSoh:1975}, 
in which they showed that the structure of the super Lie algebra in theories obeying certain basic requirements is tightly constrained: in the massive case, for example, internal and Poincar\'e symmetries commute and the fermionic charges must commute with translations and transform as rank-$1$ spinors under the Lorentz group.  
  
The CM theorem concerns a particular spacetime of high symmetry. For a generic 
spacetime with trivial isometry group, it is obvious that the internal and 
geometric symmetries combine as a direct product, and one might think that the 
CM theorem has nothing to say except for spacetimes of high symmetry 
(see~\cite{CostaMorisson:2016} for a recent CM analogue in de Sitter spacetime).
However, the viewpoint of locally covariant QFT suggests a different approach. Rather
than focus on particular spacetimes, we will prove a result (Theorem~\ref{thm:CM}) 
that applies to the theory across \emph{all} spacetimes, and is expressed in 
terms of properties of the corresponding functor. We caution that our result 
should not be viewed as a direct generalization of the CM theorem, but 
nonetheless maintain that it is a natural analogue thereof in the context of 
locally covariant QFT. Theorem~\ref{thm:CM} shares with the CM theorem an 
emphasis on dynamics, but its method of proof is quite
different, and the statement differs from the CM theorem in important respects:
notably, it is valid in all spacetime dimensions $n\ge 2$ and it is not assumed 
that the QFT in question is interacting -- whereas there are well-known free 
theories and two-dimensional models that evade CM.  We comment more on these 
points below after first explaining the main ideas of our approach. 

It is necessary to recall two ways in which symmetry can be exhibited by a locally
covariant theory $\Af:\BkGrnd\to\Phys$. First, the \emph{spacetime symmetries} of 
a spacetime $\Mb$ are just the automorphisms $\psi:\Mb\to\Mb$ in $\BkGrnd$. Any
such automorphism is mapped automatically to an automorphism $\Af(\psi)$ of the physical system $\Af(\Mb)$ of the theory on $\Mb$, and for two such symmetries one has $\Af(\psi)\circ\Af(\varphi)=\Af(\psi\circ\varphi)$ by functoriality. In this way, the 
(generically trivial) group $\Aut(\Mb)$ of spacetime symmetries of $\Mb$ is represented in the automorphism group of $\Af(\Mb)$. Second, the \emph{internal symmetries} of the theory have a natural description. Any functor $\Af$ has an associated group, $\Aut(\Af)$, consisting of all natural isomorphisms of $\Af$ to itself. In locally covariant QFT, $\Aut(\Af)$ is the global gauge group of the theory~\cite{Fewster:gauge}. It follows 
from the definition that internal symmetries commute with spacetime symmetries. 
For this reason we will focus on their combination with the Lorentz group.

In order to give the Lorentz group some purchase in curved spacetimes, Theorem~\ref{thm:CM} is formulated for locally covariant theories defined on $\BkGrnd=\FLoc$, the
category of all $n$-dimensional globally hyperbolic
spacetimes equipped with a global coframe $e=(e^\mu)_{\mu=0}^{n-1}$ for the metric $g=\eta_{\mu\nu}e^\mu\otimes e^\nu$. 
Among other requirements, a $\FLoc$-morphism $\psi$ between spacetimes with frames $e$ and $e'$ obeys $\psi^*e'{}^\mu=e^\mu$ (see section~\ref{sec:LocFLoc}). The category $\FLoc$ provides a minimal setting for general locally covariant theories and was introduced in order to discuss the spin-statistics connection~\cite{Few_Regensburg:2015,Fewster:MG2015}; it has also found use in the perturbative programme~\cite{Rejzner_book}. For our purposes the key point is that the restricted
Lorentz group $\Lc_0$ acts on $\FLoc$, by modifying the coframe 
as $e\mapsto \Lambda e$, where $(\Lambda e)^\mu= \Lambda^\mu_{\phantom{\mu}\nu}e^\nu$.
This group action leaves the metric and (time)-orientation unchanged, and
physical theories should be covariant with respect to it.

Lorentz covariance in this sense is neither an internal nor a spacetime symmetry
(indeed, it maps between backgrounds that are not generally linked by any morphism
of $\FLoc$). A similar situation occurs for rigid scaling, which 
also acts on $\FLoc$ and the category $\Loc$ often used in locally covariant QFT;
not all theories display rigid scale covariance, but it is useful
to be able to distinguish and analyze those that do. We therefore make a systematic analysis of theories that are covariant under a group action on $\BkGrnd$ (section~\ref{sec:WGA}) and illustrate it using rigid scaling (section~\ref{sec:scaling}) before passing to the discussion of Lorentz symmetry and our main result (section~\ref{sec:CM}).  

The outline is as follows. Suppose a group $G$ acts functorially on the category $\BkGrnd$ so that $g\in G$ maps any spacetime $\Mb$ to some $\act{g}{\Mb}$ and each morphism $\psi:\Mb\to\Nb$
to some $\act{g}{\psi}:\act{g}{\Mb}\to\act{g}{\Nb}$, with the identity acting
trivially and $\act{gh}{\Mb}=\act{g}{(\act{h}{\Mb})}$, $\act{gh}{\psi}=\act{g}{(\act{h}{\psi})}$. Given a theory
$\Af:\BkGrnd\to\Phys$, each element $g\in G$ determines
a new theory $\act{g}{\Af}$ obtained by defining
$\act{g}{\Af}(\Mb)=\Af(\act{g}{\Mb})$ and  $\act{g}{\Af}(\psi)=\Af(\act{g}{\psi})$.\footnote{
This action is written contravariantly, $\act{gh}{\Af}=\act{h}{(\act{g}{\Af})}$, to avoid a proliferation of inverses.} We say that $\Af$ is $G$-covariant if all these theories are physically equivalent,
meaning that there is a natural isomorphism between $\Af$ and each $\act{g}{\Af}$.
As will be shown, these isomorphisms determine a group $2$-cocycle
of $G$ with coefficients in the (potentially nonabelian) gauge group $\Aut(\Af)$. It turns out (theorem~\ref{thm:cohom}) that this $2$-cocycle is intrinsic to $\Af$; any other system of isomorphisms between $\Af$ and the $\act{g}{\Af}$ results in a cohomologous $2$-cocycle; in other words the $G$-covariance determines a distinguished cohomology class $[\Af]_G\in H^2(G,\Aut(\Af))$. Associated with this class is a canonical group extension $E$ of $G$ by $\Aut(\Af)$, 
under which the fields of the theory transform in multiplets (theorem~\ref{thm:grouprep}). 
A key question is whether such $E$-multiplets might contain inequivalent submultiplets for the action of $G$ that are mixed under the action of $E$.
This can be excluded (for irreducible $G$-multiplets) if $E$ is simply a direct product $E=\Aut(\Af)\times G$, which holds if $[\Af]_G$ is trivial (corollary~\ref{cor:nomix}). 

Theorem~\ref{thm:CM} uses this general analysis to prove that any theory 
defined on $\FLoc$ obeying the timeslice property, additivity and
dynamical local Lorentz invariance is $\Sc$-covariant with a trivial 
$2$-cocycle, where $\Sc$ is the universal covering of the 
restricted Lorentz group. These conditions will be described in 
detail later; the first two are standard and express the existence of dynamics 
and the ability to build up the theory from subspacetimes (as expected for a 
theory of quantum fields). The third uses relative Cauchy 
evolution~\cite{BrFrVe03}, the dynamical response to perturbations in the 
background structures, to express invariance with respect to local changes of 
frame. Theorem~\ref{thm:CM} is proved by 
an explicit geometrical construction, using  
smooth deformations of the background frame to connect a given framed spacetime
$\Mb$ to $\act{\Lambda}{\Mb}$, which differs from $\Mb$ only by a rigid Lorentz frame rotation. The timeslice property induces an isomorphism
between $\Af(\Mb)$ and $\Af(\act{\Lambda}{\Mb})$ which depends on the 
deformation only via its homotopy class (as a result of local dynamical Lorentz invariance)
so the covering group $\Sc$ enters in a manner reminiscent of Dirac's belt trick. 
One then shows these individual isomorphisms implement $\Sc$-covariance with trivial $2$-cocycle.
As a consequence, the extended group
is a direct product $E=\Aut(\Af)\times\Sc$, and all fields of the theory transform in 
multiplets under true representations of $\Sc$. Further, the possibility of noninteger spin can be related to the structure of the centre of the global gauge group. Thus, a theory of observables alone, with trivial global gauge group, can only admit integer spin; the same is true, for different reasons, of any theory initially defined on the category $\Loc$ of globally hyperbolic spacetimes.  

We have mentioned that Theorem~\ref{thm:CM} drops some crucial assumptions of the CM theorem. 
For example, the CM theorem requires interaction because some free Minkowski theories have symmetries that mix fields of different spin. Theorem~\ref{thm:CM} replaces this  by 
the assumption that the theory can be formulated in all spacetimes in a locally covariant fashion and that the symmetries under discussion are present in general
spacetimes. To illustrate the point, consider free scalar and Proca fields $\phi$ and $A$ with equal nonzero mass in $n=4$ Minkowski space. The current $j_{ab}=\phi\stackrel{\leftrightarrow}{\partial}_a A_b$ is conserved 
on-shell and generates a group action that mixes $\phi$ 
and $A$ in a nonlocal fashion~\cite[\S 5]{Lopus:1971}. However, this symmetry
does not extend to curved spacetimes\footnote{Replacing $\partial_a$ by covariant derivatives, $\nabla^a j_{ab} = -\phi R_{bc}A^c$ on-shell, for example; in general there is no conserved rank-$2$ combination of $\phi$ and $A$ and their derivatives.} and so there is no contradiction with Theorem~\ref{thm:CM}:
from a curved spacetime perspective, this higher spin symmetry is a quirk of the 
vacuum representation of the Minkowski theory. Similar remarks apply to  
factorizing models in $n=2$ Minkowski space that evade the CM theorem~\cite{Parke:1980}. Further comments and extensions are discussed in section~\ref{sec:conclusion}. 


\section{$G$-covariance}\label{sec:WGA}

\subsection{Motivating examples}\label{sec:LocFLoc} 

Three  categories of spacetimes will be needed: $\Loc$, $\FLoc$ and $\SpinLoc$. 
$\Loc$ is the category of oriented globally hyperbolic spacetimes~\cite{BrFrVe03}
with objects $\Mb=(\Mc,g,\ogth,\tgth)$ comprising a smooth paracompact manifold $\Mc$
of fixed dimension $n\ge 2$ and at most finitely many components, 
a smooth Lorentzian metric $g$ on $\Mc$ with signature $+-\cdots -$,
and an orientation $\ogth$ and time-orientation $\tgth$ represented
as equivalence classes of nonvanishing $n$-forms or time-like $1$-forms. 
It is required that $\Mb$ be globally hyperbolic: every $J_\Mb^+(p)\cap J^-_\Mb(q)$
is compact ($p,q\in\Mb$) and there are no closed timelike curves; equivalently $\Mb$ has Cauchy surfaces. Morphisms in $\Loc$ are smooth isometric embeddings, preserving the orientation and time-orientation, and with causally convex image; thus all causal relations between points in the image of a morphism are already present in its domain.  

$\FLoc$ is the category of framed globally hyperbolic spacetimes,\footnote{See \cite{Few_Regensburg:2015,Fewster:MG2015}; a related
category appears in \cite[Ch.~6]{Ferguson_PhD}.} the objects of which are all pairs $\Mbb=(\Mc,e)$, where $\Mc$ is a smooth $n$-dimensional manifold with smooth global coframe $e=(e^\nu)_{\nu=0}^{n-1}$ such that
\[
\FfL(\Mc,e):= (\Mc, \eta_{\mu\nu} e^\mu \otimes e^\nu, [e^0\wedge\cdots\wedge e^{n-1}],[e^0])
\]
defines an object of $\Loc$. Here $\eta_{\mu\nu} e^\mu \otimes e^\nu$ is the 
\emph{$e$-metric}, where $\eta=\text{diag}(+1,-1,\cmpbreak\ldots,-1)$, $[e^0]$ is the
equivalence class of nonvanishing $e$-timelike covector fields containing $e^0$, and $[e^0\wedge\cdots\wedge e^{n-1}]$ is the equivalence class of nonvanishing $n$-forms containing $e^0\wedge\cdots \wedge e^{n-1}$. Thus we
form the spacetime metric and (time-)orientation from the coframe
and require the resulting structure to be globally hyperbolic.  
A morphism $\psi:(\Mc,e)\to (\Mc',e')$ in $\FLoc$ is determined by a
smooth map $\psi:\Mc\to\Mc'$ that induces a $\Loc$-morphism $\FfL(\Mc,e)\to \FfL(\Mc',e')$ and obeys $\psi^*e' = e$. In this way, 
$\FfL$ is promoted to a functor $\FfL:\FLoc\to\Loc$.

Finally, $\SpinLoc$ is the category of globally hyperbolic spacetimes with spin structure, restricting to those for which the spin bundle is trivial (which includes all orientable globally hyperbolic
spacetimes in $n=4$ dimensions~\cite{Isham_spinor:1978}).
Let $\Sc$ be the universal cover of the restricted Lorentz group $\Lc_0=\SO_0(1,n-1)$,
with covering homomorphism $\pi:\Sc\to\Lc_0$.  In brief,\footnote{
The presentation here is streamlined and will be described in detail elsewhere~\cite{Few_spinstats}.} the objects of $\SpinLoc$ are exactly those of $\FLoc$, but a $\SpinLoc$ morphism from $(\Mc,e)$ to $(\Mc',e')$ is a pair $(\psi,\Xi)$ where the $\Loc$-morphism $\psi:\FfL(\Mc,e)\to\FfL(\Mc',e')$ and map $\Xi\in C^\infty(\Mc,\Sc)$ obey 
$\psi^*e' = \pi(\Xi)e$. Composition of morphisms is given by 
$(\psi',\Xi')\circ (\psi,\Xi)= (\psi'\circ\psi,(\psi^*\Xi')\Xi)$, 
where 
$((\psi^*\Xi')\Xi)(p)=\Xi'(\psi(p))\Xi(p)$. 

There is a functor $\FfS:\FLoc\to\SpinLoc$ given by $\FfS(\Mbb)=\Mbb$,
$\FfS(\psi)=(\psi,1)$, and a functor $\Uf:\SpinLoc\to\Loc$ given by $\Uf(\Mbb)=
\FfL(\Mbb)$, $\Uf(\psi,\Xi)=\psi$, with composition $\Uf\circ \FfS=\FfL$. 
Therefore any theory $\Af$ on $\Loc$ induces theories $\Af\circ\Uf$ on $\SpinLoc$ 
and $\Af\circ\FfL$ on $\FLoc$, while any theory $\Bf$ on $\SpinLoc$
(e.g., the Dirac field~\cite{Sanders_dirac:2010}) induces a theory $\Bf\circ\FfS$ on $\FLoc$.  

The category $\FLoc$ has a number of advantages: it is an operationally motivated arena for curved spacetime physics in which measurements are made with respect to a system of rods and clocks. Unlike $\Loc$, it admits theories of both integer and noninteger spin; unlike $\SpinLoc$, the objects and morphisms are given entirely in terms of observable structures. 

All three categories admit physically relevant group actions:
\begin{example}\label{ex:scaling}  
The multiplicative group $\RR^+$ acts on $\Loc$ by rigid metric scaling:
for each $\lambda\in\RR^+$,  there is a functor $\Rf(\lambda):\Loc\to\Loc$ defined on objects by 
\[
\Rf(\lambda)(\Mc,g,\ogth,\tgth) = (\Mc,\lambda^2 g,\ogth,\tgth)
\]
and so that $\Rf(\lambda)(\psi)$ has the same underlying map of manifolds as $\psi$
for any morphism $\psi$ of $\Loc$. The length of a curve in $\Rf(\lambda)(\Mb)$ 
is $\lambda$ times its length in $\Mb$; alternatively, one may think of $\Rf(\lambda)(\Mb)$
as a version of $\Mb$ in which the fundamental unit of length has been divided by $\lambda$. Given a theory $\Af:\Loc\to\Phys$
we obtain a new theory $\Af\circ\Rf(\lambda)$ for each $\lambda\in\RR^+$; the theory is (rigidly) scale covariant if all these theories are equivalent, i.e., naturally isomorphic functors ---
see section~\ref{sec:scaling} for a specific example. Of course
scaling acts in similar ways on both $\FLoc$ and $\SpinLoc$.  
\end{example}
\begin{example}\label{ex:Lorentz}
The Lorentz group $\Lc$ acts functorially on $\FLoc$ by $\Tf(\Lambda)(\Mc,e) = (\Mc,\Lambda e)$,  where $(\Lambda e)^\mu=\Lambda^\mu_{\phantom{\mu}\nu}e^\nu$ is the Lorentz-transformed coframe;
the action of $\Tf(\Lambda)$ on morphisms is defined so as to preserve the underlying map of manifolds. It is clear that $\Tf(\Lambda'\Lambda)=\Tf(\Lambda')\circ\Tf(\Lambda)$. In the present paper we only
consider the action of the restricted Lorentz group $\Lc_0$ (i.e., the identity component of $\Lc$) for which $\FfL(\Tf(\Lambda))$ is the
identity; the discrete transformations will be discussed elsewhere. 
A theory $\Af:\FLoc\to\Phys$ is (rigidly) Lorentz covariant if 
$\Af$ and $\act{\Lambda}{\Af}:=\Af\circ\Tf(\Lambda)$ are equivalent for all $\Lambda\in\Lc_0$.  
\end{example}
\begin{example}\label{ex:spin}
The universal cover $\Sc$ of $\Lc_0$ acts on $\FLoc$ by means of $\Tf\circ\pi$.
It also acts on $\SpinLoc$, by means of functors $\Sc(S)$ agreeing with $(\Tf\circ\pi)(S)$ on objects and giving $\Sc(S)(\psi,\Xi)=(\psi,S\Xi S^{-1})$ on morphisms. All theories
$\Af:\SpinLoc\to\Phys$ are $\Sc$-covariant via $\Sf$: to each $S\in\Sc$ there is a
natural isomorphism $\eta(S):\Af\nto\act{S}{\Af}$ with components $\eta(S)_\Mbb=\Af(\id_{\FfL(\Mbb)},S)$,
as shown by the calculation
\[
\eta(S)_{\Mbb'}\Af(\psi,\Xi)=\Af(\psi,S\Xi) = \Af(\psi,S\Xi S^{-1})\Af(\id_{\FfL(\Mbb)},S)=
\act{S}{\Af}(\psi,\Xi)\eta_\Mbb(S)
\]
for any $(\psi,\Xi):\Mbb\to\Mbb'$. The corresponding $2$-cocycle is trivial (see below).
\end{example}

\subsection{General analysis} The examples above motivate the
study of the following situation. Let $G$ be any group and suppose there is a  homomorphism $\Tf:G\to\Aut(\Ct)$, where $\Ct$ is a category and $\Aut(\Ct)$ is the group of invertible functors from $\Ct$ to itself.  Clearly $\Tf(g)$ has inverse $\Tf(g^{-1})$ and so every 
morphism of $\Ct$ is contained
in the image of each $\Tf(g)$. For brevity, we will often write the action of $\Tf(g)$
on objects $C$ and morphisms $\gamma$ of $\Ct$ by 
$\act{g}{C}:=\Tf(g)(C)$, and
$\act{g}{\gamma}:=\Tf(g)(\gamma)$.

\begin{definition} \label{def:weakaction}
A functor $\Af:\Ct\to\Ct'$ is \emph{$G$-covariant\footnote{There is an
unhappy collision of terminology: $\Af$ is a covariant functor in the usual category theory sense; $G$-covariance is an additional and somewhat different property.} 
(via $\Tf$)} if 
all the functors $\act{g}{\Af}=\Af\circ\Tf(g)$ are naturally isomorphic;
any family $\eta(g):\Af\nto\act{g}{\Af}$ ($g\in G$) of natural isomorphisms 
with $\eta(1)=\id_\Af$ is an \emph{implementation} of the $G$-covariance. 
\end{definition} 
Here $\Ct'$ is any category. It will be shown that
all implementations of a $G$-covariance are equivalent in the sense of nonabelian cohomology, and correspond to a uniquely determined element of the second cohomology set $H^2(G,\Aut(\Af))$. 
 
Let us briefly recall that if $G$ and $A$ are (not necessarily abelian) groups then a $2$-cochain of $G$ with coefficients in $A$ is a pair $(\xi,\phi)$ consisting of maps $\xi:G\times G\to A$ and $\phi:G\to\Aut(A)$; 
$(\xi,\phi)$ is a $2$-cocycle if  
\begin{align}
\phi(g')\phi(g)\phi(g'g)^{-1} &= \ad(\xi(g',g)) & (g',g\in G),\label{eq:cocycle_phi}\\
\xi(g'',g')\xi(g''g',g) &=\phi(g'')(\xi(g',g)) \xi(g'',g'g)   & (g'',g',g\in G) \label{eq:cocycle_xi}
\end{align} 
and the set of such $2$-cocycles is denoted $Z^2(G,A)$. Two $2$-cocycles $(\xi,\phi),(\xit,\phit)\in Z^2(G,A)$ are cohomologous precisely if there is a map $\zeta:G\to A$ 
such that 
\begin{equation}\label{eq:cohom_phixi}
\phit(g)=\ad(\zeta(g))\circ\phi(g) \quad\text{and}\quad 
\xit(g',g) = \zeta(g')\phi(g')(\zeta(g))\xi(g',g)\zeta(g'g)^{-1} 
\end{equation}
for all $g',g\in G$. The corresponding equivalence classes form the cohomology set $H^2(G,A)$, with the class of the trivial $2$-cocycle
$(1_A,\id_A)$ as a distinguished element making $H^2(G,A)$
a pointed set. Here $1_A(g',g)=1\in A$ for all $g',g\in G$. Cocycles of the form $(1_A,\phi)$, where $\phi$ is (necessarily) a homomorphism
are called \emph{neutral}, as are the corresponding cohomology classes.
A $2$-cocycle $(\xi,\phi)$ is \emph{normalized} if $\phi(1)=1$ and $\xi(g,1)=\xi(1,g)=1$ for
all $g\in G$. 

With these definitions established, our first result is:
\begin{theorem}\label{thm:cocycle}
Any implementation $\eta$ of a $G$-covariance of $\Af:\Ct\to\Ct'$ determines a normalized $2$-cocycle
$(\xi,\phi)\in Z^2(G,\Aut(\Af))$ given by
\begin{align}
\xi(g',g)_{\act{g'g}{C}}&=\eta(g')_{\act{g}{C}}\eta(g)_C\eta(g'g)^{-1}_C
& ( g',g\in G, C\in\Ct) \label{eq:xigg}\\
\phi(g)(\alpha)_{\act{g}{C}}&=\eta(g)_C\alpha_C\eta(g)^{-1}_C & (\alpha\in\Aut(\Af), g\in G, C\in\Ct).\label{eq:phig}
\end{align}  
\end{theorem}
\begin{proof}
Eqs.~\eqref{eq:xigg} and \eqref{eq:phig} are easily seen to define automorphisms $\xi(g',g)_C$ and $\phi(g)(\alpha)_C$ of $\Af(C)$ for every $C\in\Ct$ by the properties of $\Tf(g)$ described above. The rest of 
the proof is broken into several calculations.  

{\noindent\em Naturality and automorphism properties of $\phi$:} Suppose that $\gamma:C\to C'$. Then
\begin{align*}
\Af(\act{g}{\gamma})\phi(g)(\alpha)_{\act{g}{C}} &=
\eta(g)_{C'} \Af(\gamma)\alpha_C\eta(g)_C^{-1} =
\eta(g)_{C'} \alpha_{C'}\Af(\gamma)\eta(g)_C^{-1} \\
&=
\eta(g)_{C'} \alpha_{C'}\eta(g)_{C'}^{-1}\Af(\act{g}{\gamma})
= \phi(g)(\alpha)_{\act{g}{C'}}\Af(\act{g}{\gamma})
\end{align*}
so each $\phi(g)(\alpha)\in\Aut(\Af)$. It is clear from \eqref{eq:phig}
that $\phi(g)(\alpha\beta)=\phi(g)(\alpha)\phi(g)(\beta)$ so
$\phi:g\to\phi(g)$ is a map from $G$ to $\Aut(\Aut(\Af))$. 

{\noindent\em Naturality of $\xi(g',g)$:} This is proved by calculating, for 
arbitrary $\gamma:C\to C'$, 
\begin{align*}
\xi(g',g)_{\act{g'g}{C'}}\Af(\act{g'g}{\gamma})
&=\eta(g')_{\act{g}{C'}}\eta(g)_{C'}\eta(g'g)^{-1}_{C'}\Af(\act{g'g}{\gamma}) 
\\&
=\eta(g')_{\act{g}{C'}}\eta(g)_{C'}\Af(\gamma)\eta(g'g)^{-1}_{C} \\
&=\eta(g')_{\act{g}{C'}}\Af(\act{g}{\gamma})\eta(g)_{C}\eta(g'g)^{-1}_{C}\\
&=\Af(\act{g'g}{\gamma}) \eta(g')_{\act{g}{C}}\eta(g)_{C}\eta(g'g)^{-1}_{C} \\
&= \Af(\act{g'g}{\gamma}) \xi(g',g)_{\act{g'g}{C}}.
\end{align*}

{\noindent\em Cocycle property:} Normalization of $\phi$
is obvious from \eqref{eq:phig}; $\xi(g,1)=\xi(1,g)=1$  
is immediate using $\eta(1)=\id_{\Af}$. Let $\alpha\in\Aut(\Af)$
and compute
\begin{align*}
\ad(\xi(g',g))(\alpha)_{\act{{g'g}}{C}} &= \xi(g',g)_{\act{{g'g}}{C}}\alpha_{\act{{g'g}}{C}}\xi(g',g)_{\act{{g'g}}{C}}^{-1} \\
&=
\eta(g')_{\act{g}{C}}\eta(g)_C \eta(g'g)_C^{-1}\alpha_{\act{{g'g}}{C}}
\eta(g'g)_C\eta(g)_C ^{-1}\eta(g')_{\act{g}{C}}^{-1}\\
&= \phi(g')(\phi(g)(\phi(g'g)^{-1}(\alpha)))_{\act{{g'g}}{C}}
\end{align*}
for any $g',g\in G$ and $C\in\Ct$,
so $\ad\xi(g',g)=\phi(g')\phi(g)\phi(g'g)^{-1}$ as required by 
\eqref{eq:cocycle_phi}. Finally, 
let $g'',g',g\in G$ and $C\in\Ct$ be arbitrary, then
\begin{align*}
(\xi(g'',g')\xi(g''g',g))_{\act{{g''g'g}}{C}} &=
\eta(g'')_{\act{{g'g}}{C}} \eta(g')_{\act{g}{C}}   \eta(g)_C \eta(g''g'g)^{-1}_C
\\
&=
\eta(g'')_{\act{{g'g}}{C}}  \xi(g',g)_{\act{g'g}{C}}
\eta(g'')^{-1}_{\act{g'g}{C}} 
\xi(g'',g'g)_{\act{{g''g'g}}{C}} \\
&=\phi(g'')(\xi(g',g))_{\act{g''g'g}{C}}\xi(g'',g'g)_{\act{{g''g'g}}{C}}
\end{align*}
so the cocycle condition \eqref{eq:cocycle_xi} also holds. 
Thus $(\xi,\phi)\in Z^2(G,A)$. \cmpqed
\end{proof}
For example, the $2$-cocycle mentioned in Example~\ref{ex:spin}
is trivial, because $\eta(S'S)_\Mbb\cmpbreak=\arxivbreak \Af(\id_{\FfL(\act{S}{\Mbb})},S')\Af(\id_{\FfL(\Mbb)},S)
=\eta(S')_{\act{S}{\Mbb}}\eta(S)_\Mbb$, and
$\eta(S)_\Mbb\alpha_\Mbb=\alpha_{\act{S}{\Mbb}}\eta(S)_\Mbb$ by naturality of $\alpha\in\Aut(\Af)$ and the definition of $\eta(S)$.

The $2$-cocycle given by Theorem~\ref{thm:cocycle} is
intrinsic to $\Af$.
\begin{theorem} \label{thm:cohom}
If $\Af$ is $G$-covariant, the $2$-cocycles
of its implementations form a distinguished cohomology class $[\Af]_G\in H^2(G,\Aut(\Af))$. 
\end{theorem} 
\begin{proof} We show that all implementations
induce cohomologous $2$-cocycles, and all elements of
the corresponding cohomology class arise from implementations.

First, let $g:\mapsto\eta(g)$ be an implementation, let $\zeta:G\to\Aut(\Af)$ be any map and set $\ett(g)_C=\zeta(g)_{\act{g}{C}}\eta(g)_C$. Then
$g\mapsto \ett(g)$ also implements the $G$-covariance,
and  $\eta$ and $\ett$ define cohomologous $2$-cocycles. To see this, note that
each $\ett(g)_C:\Af(C)\to\Af(\act{g}{C})$ is certainly an isomorphism. If $\gamma:C\to C'$ then
\begin{align*}
\ett(g)_{C'}\Af(\gamma) &= \zeta(g)_{\act{g}{C'}}\eta(g)_{C'} \Af(\gamma) 
= \zeta(g)_{\act{g}{C'}}\Af(\act{g}{\gamma}) \eta(g)_C 
=
\Af(\act{g}{\gamma})  \zeta(g)_{\act{g}{C}}  \eta(g)_C \\&=
\Af(\act{g}{\gamma})\ett(g)_{C},
\end{align*}
which establishes naturality, so $g\mapsto \ett(g)$ implements the $G$-covariance. 
The corresponding $2$-cocycle $(\xit,\phit)$ is computed as follows:
\begin{align*}
\phit(g)(\alpha)_{\act{g}{C}} &= 
\zeta(g)_{\act{g}{C}}\eta(g)_C \alpha_C \eta(g)^{-1}_C \zeta(g)^{-1}_{\act{g}{C}} = (\ad \zeta(g))(\phi(g)(\alpha))_{\act{g}{C}},
\end{align*}
while
\begin{align*}
\xit(g',g)_{\act{g'g}{C}} &=
\zeta(g')_{\act{g'g}{C}} \eta(g')_{\act{g}{C}} \zeta(g)_{\act{g}{C}}
\eta(g)_C \eta(g'g)_C^{-1}\zeta(g'g)_{\act{g'g}{C}}^{-1} \\
&=
\zeta(g')_{\act{g'g}{C}} \phi(g')(\zeta(g))_{\act{g'g}{C}} 
\eta(g')_{\act{g}{C}}  
\eta(g)_C \eta(g'g)_C^{-1}\zeta(g'g)_{\act{g'g}{C}}^{-1} \\
&=(\zeta(g')\phi(g')(\zeta(g))\xi(g',g)\zeta(g'g)^{-1})_{\act{g'g}{C}}.
\end{align*}
The conditions in \eqref{eq:cohom_phixi}  
are met so the $2$-cocycles are cohomologous.

To prove the result, we suppose that implementations $\eta$ and
$\ett$ have been given. If the morphisms 
$\zeta(g)_{\act{g}{C}}:=
\ett(g)_C\eta(g)^{-1}_C$ form the components of an automorphism
$\zeta(g)\in\Aut(\Af)$ for each $g$, then the first part of the proof
demonstrates that the implementations induce the same cohomology
class. As the maps $\zeta(g)_{\act{g}{C}}$ are clearly isomorphisms
it remains to check naturality: if $\gamma:C\to C'$,  then
\begin{align*}
\zeta(g)_{\act{g}{C'}}\Af(\act{g}{\gamma}) &= 
\ett(g)_{C'}\eta(g)^{-1}_{C'} \Af(\act{g}{\gamma}) = 
\ett(g)_{C'}\Af(\gamma) \eta(g)^{-1}_{C}  = 
\Af(\act{g}{\gamma})  \ett(g)_{C}\eta(g)^{-1}_{C} \\&=
\Af(\act{g}{\gamma}) \zeta(g)_{\act{g}{C}},
\end{align*}
which establishes naturality as every morphism is the image of $\Tf(g)$. 
Finally, if $(\xit,\phit)\sim(\xi,\phi)$ then one has
$\zeta:G\to\Aut(\Af)$ obeying \eqref{eq:cohom_phixi}, whereupon
$\ett(g)$, defined using $\zeta$ as above, implements the $G$-covariance with cocycle $(\xit,\phit)$.\cmpqed
\end{proof}
If $[\Af]_G\in H^2(G,\Aut(\Af))$ is trivial, then one may 
choose an implementation 
corresponding to the trivial cocycle $(1_A,\id_A)$. In this case, 
one has
\begin{equation} \label{eq:eta_triv}
\eta(g)_C \alpha_C = \alpha_{\act{g}{C}}\eta(g)_C, \quad
\eta(g'g)_C = \eta(g')_{\act{g}{C}} \eta(g)_C, \qquad (g',g\in G, C\in\Ct).
\end{equation} 

Returning to the general case, 
suppose $\Af:\Ct\to\Ct'$ is $G$-covariant and
choose an implementation $g\mapsto\eta(g)$ with normalized $2$-cocycle $(\xi,\phi)\in Z^2(G,\Aut(\Af))$. The $2$-cocycle induces
a group extension of $G$ by $\Aut(\Af)$, described by
a short exact sequence of group homomorphisms
\begin{equation}\label{eq:shortexact}
1\rightarrow \Aut(\Af) \rightarrow E \stackrel{q}{\rightarrow} G \rightarrow 1,
\end{equation}
where the extension $E=\Aut(\Af)\times G$ as a set, and is equipped with the product
\begin{equation}\label{eq:Eproduct}
(a',g') (a,g)= (a'\phi(g')(a)\xi(g',g),g'g)
\end{equation}
for which $(1,1)$ is the identity. The unlabelled map $\Aut(\Af)\to E$ in \eqref{eq:shortexact} is $a\mapsto (a,1)$, and embeds $\Aut(\Af)$
as a normal subgroup of $E$, while $q(a,g)=g$ and realizes
$G$ as the quotient $G\cong E/\Aut(\Af)$. 
See, e.g., \cite{EilMac:1947b,AzcaIzqu:1995}. The group
extension is determined by the cohomology class $[\Af]_G$
up to a suitable equivalence of extensions. Some familiar 
cases arise as follows: the trivial cocycle gives the direct
product $\Aut(\Af)\times G$; a neutral cocycle $(1,\phi)$ gives the
semidirect product $\Aut(\Af)\rtimes_\phi G$; if $\Aut(\Af)$
is abelian then $(\xi,1)$ gives a central extension. 

A $G$-covariant theory is also covariant under the corresponding
group extension, which almost trivialises the cocycle (one might say that
it is neutralised). 
\begin{theorem} 
If $\Af:\Ct\to\Ct'$ is $G$-covariant via $\Tf$, 
then $\Af$ is $E$-covariant via $\Tf\circ q$, with
a neutral cocycle in $Z^2(E,\Aut(\Af))$.  
\end{theorem}
\begin{proof}[\ifnonCMP Proof (Sketch)\else Sketch\fi]
The $E$-covariance is implemented by $(\alpha,g)\mapsto\rho(\alpha,g)$, 
where $\rho(\alpha,g)_C=\arxivbreak\alpha_{\Tf(g)(C)}\eta(g)_C$. The $2$-cocycle
is $(1,\varphi)$, with $\varphi(\alpha,g) = \ad  \alpha\circ\phi(g).$ \cmpqed
\end{proof}
%
%

\subsection{Multiplets of locally covariant fields for $G$-covariant theories}\label{sec:multiplets} 
Consider a locally covariant QFT given as a functor $\Af:\BkGrnd\to\Alg$, where
$\BkGrnd$ is $(\sf Spin)\Loc$ or $\FLoc$, and $\Alg$ is the category
of unital $*$-algebras and unit-preserving $*$-monomorphisms. 
Let $\Df:\BkGrnd\to\Set$ be the functor assigning to each $C\in\BkGrnd$ the set of smooth complex-valued  compactly supported test functions on the underlying manifold of $C$, and to each morphism $\psi$, the corresponding push-forward $\Df(\psi)=\psi_*$. Let 
$\Vf$ be the forgetful functor $\Vf:\Alg\to\Set$. 
By definition, a \emph{locally covariant quantum field}~\cite{Verch01,Ho&Wa01,BrFrVe03} is a natural transformation
$\Phi:\Df\nto\Vf\circ\Af$ and the set of all such fields $\Fld(\Af)$ forms a unital $*$-algebra in a natural way~\cite{Fewster2007}: e.g.,
$(\mu\Phi+\nu\Psi)_C(f) = \mu\Phi_C(f)+\nu\Psi_C(f)$ and $(\Phi\Psi)_C(f)=\Phi_C(f)\Psi_C(f)$
($\mu,\nu\in\CC,f\in\Df(C)$) define the linear combination and product
of $\Phi,\Psi\in\Fld(\Af)$. The unit field is $\II_C(f) = \II_{\Af(\Mb)}$ for all $f\in\Df(C)$.\footnote{Using $\Set$
allows for fields depending nonlinearly on the test function. Using the category of vector spaces instead, one obtains a vector space (rather than $*$-algebra) of linear fields.} 

An advantage of theories defined on $\FLoc$ is that one need only consider single-component fields in $\Fld(\Af)$, whereas on $(\sf Spin)\Loc$ one requires a different functor $\Df$ for each tensorial field type. For example, a Proca field theory on $\Loc$ describes a field
smeared against test one-forms, $A_\Mb(\omega)$, whereas the same theory pulled back
to $\FLoc$ has available four single-component fields $A^\mu$, given by $A^\mu_{(\Mc,e)}(f):=A_{\FfL(\Mc,e)}(f e^\mu)$. The same can be done for spinor fields on spacetimes
in $\SpinLoc$, because the spin bundle is trivial. Fully worked out examples will be given elsewhere~\cite{Few_spinstats}. Of course, it is then necessary to discern some structure on the fields, which provides a useful application of $G$-covariance.

Now suppose that a group $G$ acts functorially on $\BkGrnd$, and that both $\Af$
and $\Df$ are $G$-covariant. For simplicity we assume that $G$-covariance of $\Df$
is implemented by a family $\zeta(g)$ with trivial cocycle in $Z^2(G,\Aut(\Df))$,  i.e., 
\[
\zeta(g'g)_C=\zeta(g')_{\act{g}{C}}\zeta(g)_C, \quad\text{and}\quad
\zeta(g)_C\alpha_C=\alpha_{\act{g}{C}}\zeta(g)_C
\]
for all $g',g\in G$, $\alpha\in\Aut(\Df)$ and
$C\in\BkGrnd$. In this situation, the fields in $\Fld(\Af)$ transform under both $G$ and $\Aut(\Af)$.
\begin{theorem} \label{thm:grouprep}
Suppose $\Af:\BkGrnd\to\Alg$ and $\Df:\BkGrnd\to\Set$ are $G$-covariant and
that the $G$-covariance of $\Af$ is implemented by $\eta$, with $2$-cocycle $(\xi,\phi)$, while that of $\Df$ is implemented by $\zeta$, with trivial cocycle. Let $\Phi\in\Fld(\Af)$. Then for each $\alpha\in\Aut(\Af)$ there is a
transformed field $\alpha\cdot\Phi\in\Fld(\Af)$ defined by
\begin{equation}\label{eq:dotdef}
(\alpha\cdot \Phi)_C = \Vf(\alpha_C)\Phi_C, \qquad (C\in\BkGrnd)
\end{equation}
and for each $g\in G$ there is a transformed field $g\ast\Phi\in\Fld(\Af)$
defined by  
\begin{equation}\label{eq:gstardef}
(g\ast\Phi)_{\act{g}{C}}\zeta(g)_C = \Vf(\eta(g)_C)\Phi_C, \qquad (C\in\BkGrnd).
\end{equation} 
One has $1_{\Aut{\Af}}\cdot \Phi=\Phi = 1_G\ast \Phi$ for all $\Phi\in\Fld(\Af)$. 
The following formulae hold for all $\alpha,\beta\in\Aut(\Af)$,
$g',g\in G$ and $\Phi\in\Fld(\Af)$:
\begin{align}
\alpha\cdot(\beta\cdot\Phi) & = (\alpha\beta)\cdot\Phi \label{eq:dotdot}\\
g\ast(\alpha\cdot\Phi) &= \phi(g)(\alpha)\cdot(g\ast\Phi) \label{eq:stardot} \\ 
g'\ast(g \ast\Phi) &=
\xi(g',g)\cdot\left((g'g)\ast\Phi\right)   .
\label{eq:starcocycle}
\end{align}
$\Fld(\Af)$ carries a true group action $\rho$ of the group extension $E$ of $G$ by $\Aut(\Af)$ determined by 
$(\xi,\phi)$, given by $\rho(\alpha,g)\Phi = \alpha\cdot (g\ast\Phi)$.
\end{theorem}
\begin{proof} The statements concerning the action of $\Aut(\Af)$ are proved in~\cite[\S 3.2]{Fewster:gauge}. Turning to the action of $G$, 
we note that \eqref{eq:gstardef} defines a transformed field because
\begin{align*}
\Vf(\Af(\act{g}{\gamma})) (g\ast\Phi)_{\act{g}{C}}\zeta(g)_C &= \Vf(\Af(\act{g}{\gamma})  \eta(g)_C)\Phi_C  =\Vf(\eta(g)_{C'} \Af(\gamma))\Phi_C \\
&= \Vf(\eta(g)_{C'}) \Phi_{C'} \Df(\gamma) =(g\ast\Phi)_{\act{g}{C'}} \zeta(g)_{C'}\Df(\gamma)\\
& =(g\ast\Phi)_{\act{g}{C'}} \Df(\act{g}{\gamma}) \zeta(g)_C,
\end{align*}
for all $\gamma:C\to C'$ in $\BkGrnd$.
As the $\zeta(g)$ are isomorphisms, $g\ast\Phi\in\Fld(\Af)$. 
To prove \eqref{eq:stardot}, suppose $g\in G$ and $\alpha\in \Aut(\Af)$. Calculating
\begin{align*}
\Vf(\phi(g)(\alpha)_{\act{g}{C}})(g\ast\Phi)_{\act{g}{C}}\zeta(g)_C &= \Vf(\eta(g)_C
\alpha_{C} \eta(g)_C^{-1} \eta(g)_C)\Phi_C = 
\Vf(\eta(g)_C \alpha_C)\Phi_C  \\& = (g\ast(\alpha\cdot\Phi))_{\act{g}{C}}\zeta(g)_C,
\end{align*}
we again strip off the isomorphism $\zeta(g)_C$ to obtain the required result. Next, 
\begin{align*}
\Vf(\xi(g',g)_{\act{g'g}{C}})((g'g)\ast\Phi)_{\act{g'g}{C}} \zeta(g'g)_C &= \Vf(\xi(g',g)_{\act{g'g}{C}}\eta(g'g)_C)\Phi_C \\ &= 
\Vf(\eta(g')_{\act{g}{C}}\eta(g)_C) \Phi_C \\ 
&=\Vf(\eta(g')_{\act{g}{C}}) (g\ast\Phi)_{\act{g}{C}}\zeta(g)_C \\
&=(g'\ast (g\ast\Phi))_{\act{g'g}{C}}\zeta(g')_{\act{g}{C}}\zeta(g)_C \\
&= (g'\ast (g\ast\Phi))_{\act{g'g}{C}}\zeta(g'g)_C ,
\end{align*}
for $g',g\in G$,
using the fact that $\zeta$ induces a trivial cocycle.

The final statement follows from \eqref{eq:stardot} and \eqref{eq:starcocycle} by 
the calculation
\begin{align*}
\rho(\alpha',g')\rho(\alpha,g)\Phi &= \alpha'\cdot \left(
g'\ast\left(\alpha\cdot(g\ast\Phi)\right)\right) =
\alpha'\cdot\phi(g')(\alpha)\cdot\left(g'\ast(g\ast\Phi)\right) \\
&= \left(\alpha' \phi(g')(\alpha)\xi(g',g)\right)\cdot\left((g'g)\ast\Phi\right) = \rho((\alpha',g')(\alpha,g))\Phi.  \qquad\cmpqed
\end{align*} 
\end{proof}
For example, the component fields of a Proca field transform in 
a vector representation of $\Lc_0$, $\Lambda\ast A^\mu = (\Lambda^{-1})^\mu_{\phantom{\mu}\nu} A^\nu$. Thus they can be distinguished from the components of a Dirac spinor or  
four independent scalars. 

In general, Theorem~\ref{thm:grouprep} allows one to classify fields by the subrepresentations of $\rho$ in which they transform. A subspace of $\Fld(\Af)$ (or sometimes, a basis for it) carrying an indecomposable subrepresentation of $\rho$ will be called an \emph{$E$-multiplet}, augmenting
the description with attributes of the subrepresentation (e.g., irreducibility) as appropriate.  
The same can be done for the actions of $\Aut(\Af)$ and $G$  (in the latter case, allowing generalized multiplier representations according to  \eqref{eq:starcocycle}) and referring to $\Aut(\Af)$- and $G$-multiplets respectively. One multiplet can be contained in another, if the latter is
reducible.  Note also that if $\zeta(g)$ commutes with complex conjugation, 
then the conjugate field $\Phi^\dagger$ to $\Phi\in\Fld(\Af)$
defined by $\Phi_C^\dagger(f)=\Phi_C(\overline{f})^*$ obeys
$g\ast\Phi^\dagger=(g\ast\Phi)^\dagger$ and transforms in 
the complex conjugate representation of that in which $\Phi$ transforms. 
Thus, self-adjoint fields transform in self-conjugate multiplets. 

The general structure raises the possibility that distinct $G$-multiplets can be
mixed within a larger $E$-multiplet. This can be excluded in some circumstances: 
%
%

\begin{corollary}\label{cor:nomix}
Under the hypotheses of Theorem~\ref{thm:grouprep}, suppose
additionally that $[\Af]_G\in H^2(G,\Aut(\Af))$ is trivial, 
so $E=\Aut(\Af)\times G$. Then no inequivalent irreducible nontrivial $G$-multiplets
can be mixed by the action of $E$.  
\end{corollary} 
\begin{proof}
Let $(\sigma_i,U_i)$ ($i=1,2$) be irreducible $G$-representations arising
as $G$-multiplets, i.e., there are linear injections $\iota_{i}:U_i\to\Fld(\Af)$ and surjections $\pi_i:\Fld(\Af)\to U_i$ so that $\pi_i \rho(1,g)= \sigma_i(g)\pi_i$,  $\rho(1,g)\iota_i=\iota_i\sigma_i(g)$, and $\pi_i\iota_i=\id_{U_i}$. 
If the multiplets mix, there is $e\in E$, which can be taken
without loss in the form $e=(\alpha,1)$, so that $Q=\pi_1 \rho(e) \iota_2$
and $R=\pi_2\rho(e)\iota_1$ are not both zero. We assume $R\neq 0$ without loss,
and calculate $\sigma_1(g)Q = Q\sigma_2(g)$ and $R\sigma_1(g)=\sigma_2(g)R$,
so $\Im Q$ and $\ker R$ carry subrepresentations of $\sigma_1$, while $\ker Q$ and $\Im R$ carry subrepresentations of $\sigma_2$. By irreducibility of $\sigma_i$, $R$ has trivial kernel and
cokernel; hence it is an isomorphism giving $\sigma_1\simeq \sigma_2$, contradicting the hypothesis. \cmpqed
\end{proof} 
%
%
%
Our analysis has been purely algebraic. 
We comment further on this in section~\ref{sec:conclusion};
here we mention that,  while there are discontinuous finite dimensional
representations of many groups including $\RR^+$ and $\SL(2,\CC)$, there are also various `automatic continuity' results. For example, all locally bounded finite-dimensional representations of $\SL(2,\CC)$ are continuous in the Lie group topology~\cite{Shtern:2008}.

\section{Scaling}\label{sec:scaling}

As a first illustration we consider the theory of a massless free field with general
curvature coupling. The field equation  $(\Box+\xi R)\phi=0$ is invariant
under rigid scaling of the metric; we will show that this induces a 
$\RR^+$-covariance via the group action on $\Loc$ of Example~\ref{ex:scaling},
and that the local Wick powers transform in nontrivial multiplets.
We work in $n=4$ dimensions with $\hbar=c=1$, so $\phi$ has dimensions of inverse 
length.  For brevity, we write $\Rf(\lambda)(\Mb)=\lambda\Mb$.

\paragraph{Construction of the theory} The locally covariant description of the QFT is a functor $\Wf:\Loc\to\Alg$, where each 
$\Wf(\Mb)$ is the extended algebra of Wick polynomials~\cite{Ho&Wa01}, thereby including
the local Wick powers in $\Fld(\Wf)$. 

Some preliminaries are required: 
for each $\Mb\in\Loc$, set $P_\Mb=\Box_\Mb+\xi R_\Mb$ and let $E^{\,+\smash{/}-}_\Mb$ be the corresponding retarded/advanced Green operators obeying 
$P_\Mb E^\pm_\Mb f=f$, $\supp E^\pm_\Mb f\subset J^\pm_\Mb(\supp f)$, writing also
\[
E_\Mb(f,g) = \int_\Mb f(p)\left([E^-_\Mb-E^+_\Mb]g\right)(p) \dvol_\Mb(p).
\] 
Further, choose a $P_\Mb$-bisolution $W_\Mb\in\Df'(\Mb\times\Mb)$ obeying
\begin{itemize}
\item reality conditions, $\overline{W_\Mb(f,g)}=W_\Mb(\overline{g},\overline{f})$  
\item a commutator condition, $W_\Mb(f,g)-W_\Mb(g,f)=i E_\Mb(f,g)$ 
\item a wavefront set constraint, $\WF(W_\Mb)\subset \Vc_+(\Mb)\times\Vc_-(\Mb)$,
\end{itemize}
where $\Vc_{+/-}(\Mb)$ are the closures of the bundles of future/past-pointing causal covectors on $\Mb$.

Given these definitions, the unital $*$-algebra $\Wf(\Mb)$ can be presented in terms of its generators and relations (we will be brief, and refer the reader to e.g.~\cite{Ho&Wa01,ChiFre:2009} for details). There is a unit $1$, and the other generators are symbols ${:}\Phi^{\otimes k}{:}_{\Mb}(u)$, labelled by $k\in \NN$ and 
\[
u\in\Tf^{(k)}(\Mb):=\left\{u \in \EE'_{\text{sym}}(\Mb^{\times k}):~\WF(u)\cap \left(\Vc_+(\Mb)^{\times k}\cup \Vc_-(\Mb)^{\times k}\right)=\emptyset\right\}
\]
so that $u\mapsto {:}\Phi^{\otimes k}{:}_{\Mb}(u)$ is linear. 
Here `sym' denotes the symmetric subspace and $\EE(X)$ is the space of smooth densities on $X$, while $\Df(X)$ are smooth
compactly supported functions, so $\Df(X)$ is canonically included in $\EE'(X)$ without specifying a volume element. The symbols and their adjoints obey relations that are conveniently expressed in terms of a formal power series 
\[
\Gc_{\Mb}[f]=\II+ \sum_{k=1}^\infty \frac{i^k}{k!}
{:}\Phi^{\otimes k}{:}_{\Mb}(f^{\otimes k}) ,\qquad (f\in\Tf^{(1)}(\Mb))
\] 
with coefficients in $\Wf(\Mb)$. Writing the $\Wf(\Mb)$-product as $\star_\Mb$, the relations are:
\begin{itemize}
\item hermiticity, $\Gc_\Mb[f] = \Gc_\Mb[-\overline{f}]^*$
\item field equation, $\Gc_\Mb[P_\Mb f]=\Gc_\Mb[0]$
\item Wick's formula, $\Gc_{\Mb}[f] \star_{\Mb}\Gc_{\Mb}[g] = \Gc_{\Mb}[f+g] e^{-W_\Mb(f,g)}$,
\end{itemize}
understood as identities between formal Taylor coefficients about $f=0$ (or $f=g=0$ for Wick's formula) under the rule ${:}\Phi^{\otimes k}{:}_{\Mb}(u)=i^{-k}\langle\delta^k\Gc_\Mb/\delta f^k|_{f=0},u\rangle$, which are
furthermore required to remain valid under linearity and taking limits 
in $\Tf^{(\bullet)}(\Mb)$ with respect to a suitable topology (or pseudo-topology) -- see~\cite{DabBro:2014} and~\cite[\S 4.4.2]{Rejzner_book} for a discussion of various possible choices.  

This completes the description of the extended algebra $\Wf(\Mb)$ (by contrast, the \emph{unextended} algebra is the unital $*$-subalgebra $\Af(\Mb)$ generated by $\Phi_\Mb(f):={:}\Phi{:}_{\Mb}(f)$, for $f\in\Df(\Mb)$).  
For completeness, however, we spell out the relations in more detail. 
Hermiticity asserts ${:}\Phi^{\otimes k}{:}_{\Mb}(u)={:}\Phi^{\otimes k}{:}_{\Mb}(\overline{u})^*$ for all $u\in\Tf^{(k)}(\Mb)$,\footnote{Here we use 
$\langle\delta^k\Hc[\overline{f}]^*/\delta f^k,u\rangle=
\langle\delta^k\Hc[\overline{f}]/\delta \overline{f}^k,\overline{u}\rangle^*$.}
while Wick's formula corresponds to the relations
\begin{align*}
{:}\Phi^{\otimes k}{:}_{\Mb}(u)\star_\Mb{:}\Phi^{\otimes \ell}{:}_{\Mb}(v) &=(-i)^{k+\ell}  \left\langle\frac{\delta^k}{\delta f^k}\otimes \frac{\delta^\ell}{\delta g^\ell} 
\Gc_{\Mb}[f+g] \left.e^{-W_\Mb(f,g)}\right|_{f,g=0},u\otimes v\right\rangle \\
&=\sum_{j=0}^{\min\{k,\ell\}} {:}\Phi^{\otimes (k+\ell-2j)}{:}_{\Mb}(u\otimes_j v)
\end{align*}
for all $u\in\Tf^{(k)}(\Mb)$, $v\in\Tf^{(\ell)}(\Mb)$.
Here, $u\otimes_j v$ is the symmetrized $j$-times $W_\Mb$-contracted tensor product
given by $u\otimes_j v=j!\binom{k}{j}\binom{\ell}{j} \Sym(w)$, where $w\in \EE'(\Mb^{\times(k+\ell-2j)})$ is defined by $w(f\otimes g) = (W_\Mb^{\otimes j} v_g)(u_f)$
for $f\in\EE(\Mb^{\times(k-j)})$, $g\in\EE(\Mb^{\times(\ell-j)})$, regarding
$W_\Mb^{\otimes j}$ as a map $\Tf^{(j)}(\Mb)\to\Df'(\Mb^{\times j})$ and 
denoting $u_f(\cdot)=u(f\otimes \cdot)\in \Tf^{(j)}(\Mb)$, $v_g(\cdot)=v(g\otimes \cdot)\in \Tf^{(j)}(\Mb)$. The microlocal conditions on $W_\Mb$ and $\Tf^{(\bullet)}(\Mb)$ ensure that all this is well-defined. Lastly, combining the field equation with Wick's formula gives $\Gc_\Mb[f+P_\Mb h]=\Gc_\Mb[f]$ and 
therefore, taking one functional derivative in $h$ and the rest in $f$, yields the 
relations ${:}\Phi^{\otimes (k+1)}{:}_{\Mb}(w)=0$ for any $w$ in the closure 
of  $\Span \{\Sym (u\otimes P_\Mb v): u\in\Tf^{(k)}(\Mb),~v\in\Tf^{(1)}(\Mb)\}\subset \Tf^{(k+1)}(\Mb)$
for $k\in\NN_0$ ($\Tf^{(0)}(\Mb)=\CC$ by convention). 
The generating function evidently provides a very compact formulation of these relations and permits
efficient computation with them. 

Returning to the definition of $\Wf$ as a functor, 
to each morphism $\psi:\Mb\to\Nb$ in $\Loc$, there is a corresponding 
$\Wf(\psi):\Wf(\Mb)\to\Wf(\Nb)$ which acts on generators by
\begin{equation}\label{eq:WfGc}
\Wf(\psi) \Gc_\Mb[f] = \Gc_\Nb[\psi_*f] e^{(W_\Mb(f,f)-W_\Nb(\psi_*f,\psi_*f))/2}
\end{equation}
and extends to an $\Alg$-morphism (cf.~\cite[\S 3]{Ho&Wa01}) ultimately because $P_\Nb\psi_*f=\psi_*P_\Mb f$ for $f\in\Tf^{(1)}(\Mb)$. Although $\Wf$ depends on the choice of $W_\Mb$'s, different choices
result in equivalent theories. For our purposes we assume without loss that $W_{\lambda\Mb}(f,g) = \lambda^6
W_\Mb(f,f')$ for all $f,f'\in\Df(\Mb)$, $\lambda\in\RR^+$. 
(This is consistent with the commutator condition because $\Box_{\lambda\Mb}=\lambda^{-2}\Box_\Mb$ and $\dvol_{\lambda \Mb}=\lambda^4\dvol_\Mb$, giving 
$E_{\lambda\Mb}^\pm f = \lambda^2 E_\Mb^\pm f$ and $E^\pm_{\lambda \Mb}(f,f')  = 
\lambda^6 E^\pm_\Mb(f,f')$.)

\paragraph{Covariance under rigid scaling} 
We now show that $\Wf$ is $\RR^+$-covariant under rigid scaling, for any $\xi\in\RR$, by exhibiting natural isomorphisms $\eta(\lambda):\Wf\nto\act{\lambda}{\Wf}$ for each $\lambda\in\RR^+$, with components defined as
\[
\eta(\lambda)_\Mb {:}\Phi^{\otimes k}_\Mb{:}(u) = \lambda^{-3k}
 {:}\Phi^{\otimes k}_{\lambda\Mb}{:}(u), \qquad(u\in \Tf^{(k)}(\Mb),~\Mb\in\Loc).
\]
Equivalently, $\eta(\lambda)_\Mb \Gc_\Mb[f] = \Gc_{\lambda\Mb}[\lambda^{-3}f]$,
in which form compatibility with the relations may be verified easily.  
Hollands and Wald studied these maps in~\cite[\S 4.3]{Ho&Wa01} (notation differs) 
and showed that they are $\Alg$-isomorphisms. Naturality was not proved in~\cite{Ho&Wa01} but is easily checked: if $\psi:\Mb\to\Nb$ then
\begin{align*}
\eta(\lambda)_\Nb\Wf(\psi)\Gc_\Mb[f] &= \eta(\lambda)_\Nb\left(
\Gc_\Nb[\psi_*f]e^{(W_\Mb(f,f)-W_\Nb(\psi_*f,\psi_*f))/2}\right) \\
&=\Gc_{\lambda\Nb}[\lambda^{-3}\psi_*f]e^{(W_\Mb(f,f)-W_\Nb(\psi_*f,\psi_*f))/2} 
\\ &=\Wf(\act{\lambda}{\psi})\eta(\lambda)_\Mb\Gc_\Mb[f],
\end{align*}
using $W_{\lambda\Mb}=\lambda^6 W_\Mb$. This 
proves that $\Wf$ is $\RR^+$-covariant. It is clear that $\eta(\lambda'\lambda)_\Mb=\eta(\lambda')_{\lambda\Mb}\eta(\lambda)_\Mb$, so the corresponding $2$-cocycle takes the form $(\id,\phi)$ where 
$\phi:\RR^+\to\Aut(\Aut(\Wf))$ remains to be determined.

For illustrative purposes, we restrict to the action of $\phi$ on a subgroup of $\Aut(\Wf)$
which --- on the basis of an analysis of the $\xi=0$ unextended theory~\cite{Fewster:gauge} ---
is expected to constitute all `regular' gauge transformations. In the case $\xi\neq 0$, 
this subgroup is a $\ZZ_2$, with action defined by $\sigma_\Mb\Gc_\Mb[f]=\Gc_\Mb[\sigma f]$ ($\sigma=\pm 1)$, while if $\xi=0$ it is the nonabelian semidirect product $\ZZ_2\ltimes\RR$, with group product 
$(\sigma',\mu')(\sigma,\mu)=(\sigma'\sigma,\mu'\sigma +\mu)$ and action specified by
\begin{equation}\label{eq:gauge}
(\sigma,\mu)_\Mb \Gc_\Mb[f]=\Gc_\Mb[\sigma f] e^{i\mu\int_\Mb f\dvol_\Mb}.
\end{equation}
Here, $\mu$ has dimensions of inverse length, like $\Phi$.  One may treat the two cases together
by restricting to $\mu=0$ if $\xi\neq 0$. Noting that
\begin{align*}
\eta(\lambda)_\Mb(\sigma,\mu)_\Mb \Gc_\Mb[f] &=
\eta(\lambda)_\Mb \Gc_\Mb[\sigma f] e^{i\mu\int  f\dvol_\Mb} 
=\Gc_\Mb[\sigma  f/\lambda^{3}] e^{i\frac{\mu}{\lambda}\int f/\lambda^3\dvol_{\lambda\Mb}} \\
&=(\sigma,\mu/\lambda)_{\lambda\Mb} 
\eta(\lambda)_\Mb\Gc_\Mb[f],
\end{align*}
we have  $\phi((\sigma,\mu)) = (\sigma,\mu/\lambda)$, 
which is consistent with the dimensions of $\mu$. Thus, the $2$-cocycle
for rigid scaling is nontrivial for minimal coupling $\xi=0$, and
(at least its restriction to the regular subgroup) is trivial for $\xi\neq 0$. 

\paragraph{Action on local Wick powers} Scaling induces a group action on 
$\Fld(\Wf)$ because $\Df$ is also $\RR^+$-covariant, implemented by $\lambda\mapsto \zeta^{(\alpha)}(\lambda)$, where $\alpha\in\RR$ and $
\zeta(\lambda)_\Mb^{(\alpha)} f= \lambda^{-4\alpha}f$ ($\Df(\Mb)=\Df(\lambda\Mb)$, 
because the manifolds coincide). One may check that the corresponding
cocycle is trivial for all $\alpha$; we take $\alpha=1$, 
so fields transform as densities of weight zero,
and now drop the superscript $\alpha$. 

As suggested by the notation, the generators ${:}\Phi^{\otimes k}{:}_{\Mb}(u)\in\Wf(\Mb)$ are (distributionally) smeared $k$-multilocal fields, Wick ordered with respect to $W_\Mb$.\footnote{Indeed, if $W_\Mb$ is of positive
type, $W_\Mb(\overline{f},f)\ge 0$ for all $f\in\Df(\Mb)$, one can define a state on
$\Wf(\Mb)$ in which all such elements have vanishing expectation value.} Owing to \eqref{eq:WfGc}, they do not transform covariantly for $k>1$, because there is no choice
of $W_\Mb$ such that $(\psi\times\psi)^*W_\Nb=W_\Mb$ for all $\psi:\Mb\to\Nb$. 
However, locally covariant Wick powers can be defined as follows. First, let $H_\Mb$ be the local Hadamard bidistribution, defined near the diagonal in $\Mb\times\Mb$ by 
\[
H_\Mb(p,q) = \frac{U_\Mb(p,q)}{4\pi^2 \sigma_{\Mb+}(p,q)} + V_\Mb(p,q) \log(\sigma_{\Mb+}(p,q)/\ell^2)
\]
where $\ell$ is a fixed length scale, common to all spacetimes, and
 $\sigma_\Mb(p,p')$ is the signed squared geodesic separation
of $p$ and $p'$, with a positive sign for spacelike separation.
The subscript $+$ indicates that 
$f(\sigma_{\Mb+}(p,q))=\lim_{\epsilon\to 0+} f(\sigma_\Mb(p,q)+2i\epsilon (T_\Mb(p)-T_\Mb(q))+\epsilon^2)$, where $T_\Mb$ increases to the future; $U_\Mb$ and $V_\Mb$ are smooth, and are fixed by requiring $U_\Mb(p,p)=1$
and $(P_\Mb\otimes 1)H_\Mb(p,q)=O(\log(\sigma_\Mb(p,q)))$. At the diagonal, $W_\Mb-H_\Mb$ is continuous and $V_\Mb$ is a multiple of the Ricci scalar: $V_\Mb(p,p) = (6\xi-1) R_\Mb|_p/(96\pi^2)$ (see, e.g.~\cite{DeWBre:1960}). 

With $H_\Mb$ so defined, set $\Hc_\Mb[f]=\Gc_\Mb[f] e^{(H_\Mb(f,f)-W_\Mb(f,f))/2}$ on $f$ of sufficiently small support that $H_\Mb$ is defined on $\supp f\times \supp f$. Then 
\begin{equation}\label{eq:lcWick}
\Phi^k_\Mb(f) = \frac{1}{i^{k}}\left\langle\left. \frac{\delta^k\Hc_\Mb}{\delta h^k}\right|_{h=0},f\delta^{(k)}_\Mb\right\rangle
\end{equation}
defines a local $k$'th Wick power smeared against $f\in\Df(\Mb)$, where 
\[
(f\delta^{(k)}_\Mb)(F) = \int_\Mb \rho_\Mb(p)^{-k} F(p,\ldots,p)f(p)\dvol_\Mb(p) \qquad (F\in\EE(\Mb^{\times k}))
\]
defines $f\delta^{(k)}_\Mb\in \Tf^{(k)}(\Mb)$; here  $\rho_\Mb$ is the density induced by $\dvol_\Mb$.
 
Under scaling, the transformed field obeys $(\lambda\ast\Phi^k)_\Mb( f)= 
\eta(\lambda)_\Mb\Phi_\Mb^k(\lambda^4 f)$, given our choice of $\zeta$ (see Theorem~\ref{thm:grouprep}). Noting that
\begin{align*}
\eta(\lambda)_\Mb\Hc_\Mb[\lambda^4 f] &= \Hc_{\lambda\Mb}[\lambda f] e^{\lambda^8 H_\Mb(f,f)-\lambda^{2}H_{\lambda\Mb}(f,f)} \\
&= \Hc_{\lambda\Mb}[\lambda f ] e^{-\lambda^2\int 
\left(H_{\lambda\Mb}(p,q)-\lambda^{-2} H_\Mb(p,q)\right) f(p)f(q)\dvol_{\lambda\Mb}^{\times 2}(p,q)}
\end{align*}
and using \eqref{eq:lcWick} together with the observations that
$\lambda^4 f\delta^{(k)}_\Mb = \lambda^{4k} f\delta^{(k)}_{\lambda\Mb}$
and  $H_{\lambda\Mb}(p,p)-\lambda^{-2}H_\Mb(p,p)\arxivbreak = V_{\lambda\Mb}(p,p)\log\lambda^2$, a short calculation gives
\begin{equation}\label{eq:almost_hom}
\lambda\ast\Phi^k = \lambda^k\sum_{j=0}^{\lfloor k/2\rfloor} \frac{k!}{j!(k-2j)!}
\left(\frac{6\xi-1}{96\pi^2}\right)^j (\log\lambda^2)^j \Rc^j \Phi^{k-2j},
\end{equation}
where $\Rc^j\Phi^k\in\Fld(\Wf)$ is the field $(\Rc^j\Phi^k)_\Mb(f)=\Phi^k_\Mb(R_\Mb^j f)$. 
Aside from the special cases $k=1$ or $\xi=1/6$, in which $\lambda\ast\Phi^k=\lambda^k\Phi^k$,
all Wick powers obey `almost homogeneous scaling'~\cite{Ho&Wa01}, and each $\Phi^k$ ($k\ge 2$) belongs to a $\lfloor k/2\rfloor$-dimensional indecomposable (and reducible) $\RR^+$-multiplet. Wick powers can be redefined within certain parameters \cite{Ho&Wa01}, but homogeneous scaling cannot be regained: for example, it is possible to redefine $\Phi^2$ by adding a fixed multiple of the Ricci scalar, but this still transforms inhomogeneously. 
We emphasise that, nonetheless, the theory $\Wf$ has rigid scale covariance for all $\xi\in\RR$.

The above discussion can be compared with~\cite{Pinamonti:2009}, which considered
theories defined on a category $\CLoc$ that admits conformal
isometries as morphisms. Only the $\xi=1/6$ conformally coupled version of
$\Wf$ is defined on $\CLoc$ and only locally conformally covariant fields can be discussed
in that setting (these include Wick powers, related to those given above within the allowed renormalization freedoms). Our approach allows us to examine
a broader class of theories that are scale covariant alongside theories that are not. 
By including the mass-squared
parameter into the background category one can even discuss theories with mass
(here the background objects are pairs $(\Mb,m^2)$ and $\RR^+$ acts by
$\Rf(\lambda)(\Mb,m^2)=(\lambda\Mb,m^2/\lambda^2)$). Elsewhere, 
it is hoped to explore  
the St\"uckelberg--Petermann renormalization group~\cite{BruDueFre2009} in our framework.

%

Summarising, this example demonstrates the need for a cohomological
description of $G$-covar{\-}iance using nonabelian coefficients, the possibility of a nontrivial action of the group $G$ ($\RR^+$ for us) on the global gauge group
and the possibility that fields can arise as indecomposable (but reducible) multiplets.

\section{An analogue of the Coleman--Mandula theorem}\label{sec:CM} 

\subsection{Hypotheses, statement of main result and consequences} 

The purpose of this section is to prove Theorem~\ref{thm:CM}, which shows that any theory $\Af:\FLoc\to \Phys$ obeying mild conditions is covariant with respect to the universal covering
group $\Sc$ of the restricted Lorentz group $\Lc_0$ (i.e., $\Sc\cong\SL(2,\CC)$ in $4$
spacetime dimensions) and has trivial cohomology class. Accordingly, the
Lorentz and internal symmetry groups do not mix, and the fields appear in  
$\Sc$ multiplets (if $\Phys=\Alg$, for example). Further consequences are discussed below.

To start, let us note that if $\Bf:\Loc\to\Alg$, then $\Af:=\Bf\circ\FfL:\FLoc\to\Alg$ is certainly $\Lc_0$-covariant, because $\FfL(\act{\Lambda}\Mbb)=\FfL(\Mbb)$ and
$\FfL(\act{\Lambda}{\psi})=\FfL(\psi)$ for all $\Mbb\in\FLoc$ and all $\psi:\Mbb\to\Nbb$. 
Thus $\act{\Lambda}{\Af}=\Af$ for all $\Lambda\in\Lc_0$, so the $\Lc_0$-covariance is implemented by $\Lambda\mapsto\id_\Af$.  
The corresponding $2$-cocycle is obviously trivial, and one obtains in a similar
way that $\Af$ is $\Sc$-covariant with trivial $2$-cocycle. 
We have already shown that any theory $\Af:\SpinLoc\to\Alg$ is
$\Sc$-covariant with trivial $2$-cocycle.\footnote{It follows that $\Af\circ\FfS$ is $\Sc$-covariant with neutral cocycle $(1,\phi)$, where $\phi$ is trivial on $\FfS^*(\Aut(\Af))$, which could \emph{a priori} be a proper subgroup of $\Aut(\Af\circ\FfS)$.}  The purpose of 
Theorem~\ref{thm:CM} is not to describe these cases as such, but rather to show 
why \emph{all} theories on $\FLoc$ obeying our conditions, however constructed, have a trivial cocycle for a common reason.  
We now proceed to assemble the hypotheses
and concepts required in Theorem~\ref{thm:CM}.

\paragraph{Timeslice property}
Given $\Mbb=(\Mc,e)\in\FLoc$, we will say that a set $\Sigma\subset\Mc$
is a Cauchy surface if it is intersected exactly once by every $e$-timelike curve. 
A morphism $\psi:\Mbb\to\Mbb'$ is said to be Cauchy
if the image of $\psi$ contains a Cauchy
surface of $\Mbb'$. Thus a $\FLoc$-morphism $\psi$ is Cauchy if and only if $\FfL(\psi)$ is Cauchy in $\Loc$ according to the terminology of~\cite{FewVer:dynloc_theory}. 
The theory $\Af$ has the \emph{timeslice
property} if $\Af(\psi)$ is an isomorphism for all Cauchy $\psi$. 
 
\paragraph{Relative Cauchy evolution \& dynamical local Lorentz invariance} 
Relative Cau{\-}chy evolution measures the response of the dynamics of 
a theory to a variation in the background structures. In $\FLoc$, variations
of $\Mbb=(\Mc,e)$ are parametrized by a smooth function 
$T\in C_\tc^\infty(\Mc;\GL^+(n;\RR))$ where the subscript indicates that
$\supp T$ (the closure of  the subset of $\Mc$ on which $T$ differs from the identity) is time-compact.  The varied coframe is $Te$, where $(Te)^{\mu}|_p=T^\mu_{\phantom{\mu}\nu}(p)e^\nu|_p$; we restrict to those $T$ for which $\Mbb[T]:=(\Mc,Te)$ is an
object of $\FLoc$. Coframe variations include, but go beyond, the metric
variations studied in~\cite{BrFrVe03,FewVer:dynloc_theory,FewVerch_aqftincst:2015} -- they can also be used to detect
whether a theory is sensitive to local Lorentz transformations. (Frame variations are required
in describing relative Cauchy evolution in the Dirac case~\cite{Sanders_dirac:2010,Ferguson_PhD}
but in an auxiliary role, whereas here they are primary.) 

Let $\Mc^\pm=I_\Mbb^\pm(\Sigma^\pm)$ where $\Sigma^\pm$
are smooth spacelike Cauchy surfaces obeying $\supp T\subset I_\Mbb^+(\Sigma^-)\cap I_\Mbb^-(\Sigma^+)$. Then 
$\Mbb^\pm=(\Mc^\pm,e|_{\Mc^\pm})$ are objects of $\FLoc$ and the subset inclusions of $\Mc^\pm$ in $\Mc$ induce Cauchy morphisms 
$\iota^\pm:\Mbb^\pm\to\Mbb$ and $\iota^\pm[T]:\Mbb^\pm\to
\Mbb[T]$. The {\em relative Cauchy evolution} $\rce_{\Mbb}[T]$
is defined by 
\[
\rce_{\Mbb}[T] = \Af(\iota^-)\Af(\iota^-[T])^{-1}\Af(\iota^+[T])\Af(\iota^+)^{-1}
\]
and is clearly an automorphism of $\Af(\Mbb)$, assuming $\Af$ has the timeslice property.
The specific choice of frame should be irrelevant in physical theories, motivating:
\begin{definition} A theory $\Af:\FLoc\to\Phys$ with the timeslice property
satisfies {\em dynamical local Lorentz invariance} if
$\rce_\Mbb[\tilde{\Lambda}]=\id$ for all $\Mbb\in\FLoc$ and all
$\tilde{\Lambda}\in C^\infty_{\tc}(\Mc;\Lc_0)$ 
that are null-homotopic relative to the complement of a time-compact
subset of $\Mc$.
\end{definition}
This condition holds in any theory induced from $\Loc$ 
or $\SpinLoc$ of the form $\Af=\Bf\circ\FfL$ or
$\Af=\Cf\circ\FfS$.\footnote{The $\Loc$ case is trivial, because $\FfL(\Mbb[T])=\FfL(\Mbb)$; 
	$\SpinLoc$ needs a short calculation.}	
Note that
the restriction to null-homotopic $\tilde{\Lambda}$ is a conservative 
assumption; a stronger definition that dropped the null-homotopy condition
would rule out theories with non-integer spin fields.\footnote{The need to
consider homotopy properties of framings in relation to relative
Cauchy evolution was noted by Ferguson~\cite{Ferguson_PhD}.}

\paragraph{Additivity} The theory $\Af$ is said to be additive 
if each $\Af(\Mbb)$ can be built from knowledge of the theory
on suitable subregions of $\Mbb$. To make this precise, 
note first that if $\Mbb=(\Mc,e)$ and $O$ is
an open $e$-causally convex subset of $\Mc$, then 
$\Mbb|_O:=(O,e|_O)$ defines the $\FLoc$-object corresponding to 
$O$ as a spacetime in its own right, and that
the inclusion of $O$ in $\Mc$ induces a $\FLoc$-morphism 
$\iota_{\Mbb;O}:\Mbb|_O\to\Mbb$.  Our additivity condition
requires that the morphisms
$\Af(\iota_{\Mbb;D})$ are jointly epic as $D$ runs over the set of truncated multi-diamonds (defined below) in $\Mb$: that is, if $\alpha\circ \Af(\iota_{\Mbb;D})=\beta\circ \Af(\iota_{\Mbb;D})$ for all truncated multi-diamonds $D$, then $\alpha=\beta$. This differs slightly
from the definition used in~\cite{FewVer:dynloc_theory,Fewster:gauge} but follows from it if
(as is true for $\Phys=\Alg$) $\Phys$ has unions and equalizers~\cite[Lem~2.5]{Fewster:gauge}.\footnote{There is a typographical error in the proof of \cite[Lem~2.5]{Fewster:gauge}; the calculation in the penultimate line should end with $h\circ m$, not $m$.} A \emph{truncated multi-diamond}
is a subset of the form $\Nc\cap D_\Mbb(B)$ where $\Nc$ is an
open globally hyperbolic neighbourhood of Cauchy surface $\Sigma$
in $\Mbb$, while the \emph{base} $B$ is a finite union of disjoint subsets of $\Sigma$
each of which is an open ball in local coordinates, and is called
a \emph{Cauchy multi-ball}.  
Images of Cauchy multi-balls under ${\sf (F)}\Loc$ morphisms are again Cauchy multi-balls.
(See Def.~2.5 and the subsequent discussion in \cite{FewVer:dynloc_theory}.) 

Given these definitions, our main result can be stated as follows. 
\begin{theorem}\label{thm:CM} In spacetime dimension $n\ge 2$, suppose $\Af:\FLoc\to\Phys$ obeys the timeslice axiom, dynamical local 
Lorentz invariance and additivity. Then $\Af$ is $\Sc$-covariant with trivial cocycle (and hence trivial cohomology class).
\end{theorem} 
Before giving the proof we make some remarks and draw out some consequences. First,
as discussed in the introduction, Theorem~\ref{thm:CM} is an analogue of the Coleman--Mandula theorem~\cite{ColeMand:1967} insofar as it is based on dynamics (specifically, the timeslice property and dynamical local Lorentz invariance), rather than on a group
theoretic analysis such as~\cite{Michel:1964,ORaif:1965,Jost:1966}. However, we re-emphasize 
that our result is not a direct generalization of the Coleman--Mandula theorem in either its
statement or its method of proof. It is also worth noting that the proof of
Theorem~\ref{thm:CM} does not utilize special properties of Minkowski spacetime, or 
of the theory $\Af$ restricted to Minkowski spacetime. In this, it differs from results such 
as the spin-statistics connection~\cite{Verch01}.

Second, triviality of the cohomology class implies that the corresponding extended symmetry
group is a direct product $E=\Aut(\Af)\times\Sc$. The single-component
fields $\Fld(\Af)$ therefore form multiplets under the action of $E$, and the restrictions of
this action to $\Aut(\Af)$ or $\Sc$ are also true representations. Thus fields arise in $\Sc$-multiplets, just as in Minkowski spacetime. By Corollary~\ref{cor:nomix},
inequivalent irreducible representations of $\Sc$ 
(or indeed of the gauge group $\Aut(\Af)$) cannot be mixed by the action of $E$, so 
finite-dimensional multiplets of different spinor-tensor type do not mix. 

Third, the proof of Theorem~\ref{thm:CM} explicitly constructs
an implementation of the $\Sc$-covariance. In Minkowski spacetime, 
this can be connected to the standard action of the Lorentz
group in Wightman theory -- see Section~\ref{sec:Mink}. 

Fourth, any $\Sc$-covariant theory is also $\Lc_0$-covariant with an implementation
given by $\Lambda\mapsto\eta(\Lambda)=\zeta(S_\Lambda)$, where
$\Lambda\mapsto S_\Lambda$ is any section of the covering homomorphism
$\pi:\Sc\to\Lc_0$ with $S_\II=1$.
The corresponding cocycle is easily calculated, using triviality of that induced by $\zeta$,
and is $(\zeta\circ z,1)$, where $z:\Lc_0\times\Lc_0\to\ker\pi$ is
given as
$z(\Lambda',\Lambda)=S_{\Lambda'}S_\Lambda S_{\Lambda'\Lambda}^{-1}$. 
The restriction of $\zeta$ to $\ker\pi$ is therefore of interest.
\begin{lemma} \label{lem:invol} $\zeta$ restricts
to a homomorphism from $\ker\pi$ to the centre $Z(\Aut(\Af))$. 
\end{lemma}
\begin{proof}
For $S\in\ker\pi$, $\act{S}{\Mbb}=\Mbb$ for each $\Mbb$ and so
$\zeta(S)\in\Aut(\Af)$. Triviality of the $\Sc$ cocycle induced by $\zeta$
(cf.~\eqref{eq:eta_triv}) implies that $\zeta|_{\ker\pi}$ is a homomorphism and that $\zeta(S)\alpha=\alpha\zeta(S)$ for all $\alpha\in\Aut(\Af)$, $S\in\ker\pi$.  \cmpqed
\end{proof}
The kernel of $\pi$ is the homotopy group $\pi_1(\Lc_0)$.   
In spacetime dimensions $n\ge 4$,  $\ker\pi \cong\ZZ_2$, 
and $\zeta(-1)$ is thus an involutive, central element of $\Aut(\Af)$, 
while in $n=3$, $\ker\pi$ is the infinite cyclic group, and in $n=2$,
it is trivial. The extended group corresponding to 
$\Lc_0$-covariance is a quotient of $\Aut(\Af)\times\Sc$; for example,
if $n\ge 4$, it is $(\Aut(\Af)\times\Sc)/\ZZ_2$, where the $\ZZ_2$
is generated by $(\zeta(-1),-1)$. As all fields transform in true $\Lc_0$-representations if $(\zeta\circ z,1)$ is trivial, one has:
\begin{corollary} If $n\ge 3$, let $\Af$ obey the conditions of Theorem~\ref{thm:CM}.  A necessary condition for $\Fld(\Af)$ to contain multiplets of noninteger spin is that $\Aut(\Af)$ carries a nontrivial homomorphic image of $\pi_1(\Lc_0)$ 
(induced by $\zeta$). In particular, $\zeta(-1)$ must
be a nontrivial involutive central element in dimension $n\ge 4$. 
\end{corollary}
The structure of $Z(\Aut(\Af))$ therefore constrains
the possible spins of fields associated with $\Af$. For example, any theory described by algebras of \emph{observables} (as opposed to possibly unobservable quantities) has trivial global 
gauge group and thus can only support multiplets of integer spin. We will return elsewhere~\cite{Few_spinstats} to the role of the \emph{univalence} $\zeta(-1)$ in the spin-statistics connection (see~\cite{Fewster:MG2015,Few_Regensburg:2015}
for brief accounts). 
Finally, it has already been noted that all theories on $\FLoc$ of the form $\Af=\Bf\circ\FfL$ are $\Lc_0$-covariant with trivial cocycle. Accordingly the fields in $\Fld(\Af)$ transform
under true $\Lc_0$-representations, proving that no theory with noninteger spin
can be constructed on $\Loc$.

\subsection{Proof of Theorem~\ref{thm:CM}}
The proof has three parts: (a) 
for $\Mbb\in\FLoc$, $S\in\Sc$, we construct
isomorphisms $\zeta_\Mbb(S):\Af(\Mbb)\to\Af(\act{S}{\Mbb})$; (b) we prove that the $\zeta_\Mbb(S)$
cohere to form natural isomorphisms and therefore implement $\Sc$-covariance of $\Af$; (c) we compute the corresponding $2$-cocycle. Additivity is used
in part (b) for reasons discussed below, while dynamical local Lorentz invariance is used
to show that $\zeta_\Mbb(S)$ is independent of various choices made in its construction,
which is important in (b) and (c). Throughout, we use the fact that 
elements of $\Sc$ can be regarded as homotopy equivalence classes of curves in $\Lc_0$ with a base-point at the identity $I$. We now take these parts in turn.

\paragraph{(a) Construction of $\zeta_\Mbb(S)$} Fix $\Mbb=(\Mc,e)$ and
$S\in\Sc$. Choose $\tilde{\Lambda}\in C^\infty(\Mc;\Lc_0)$ obeying 
$\tilde{\Lambda}\equiv I$ on  $J_{\Mbb}^-(\Sigma^-)$ and $\tilde{\Lambda}\equiv \Lambda$ on  $J_{\Mbb}^+(\Sigma^+)$, where $\Sigma^\pm$ are smooth spacelike Cauchy surfaces with $\Sigma^\pm\subset I_\Mbb^\pm(\Sigma^\mp)$; it is required that $\tilde{\Lambda}$ has homotopy class $S$ relative to $J^+_\Mbb(\Sigma^+)\cup J^-_\Mbb(\Sigma^-)$ (in every component of $\Mbb$).\footnote{In each component, every timelike curve from the past of $\Sigma^-$ to the future of
$\Sigma^+$ induces a curve connecting $I$ to $\Lambda$ in $\Lc_0$, and these curves must have common homotopy type.} Next, define $\tilde{\Mbb}=(\Mc,\tilde{\Lambda}e)$ (abusing notation,
we will sometimes write $\tilde{\Mbb}=\act{\tilde{\Lambda}}{\Mbb}$) and also $\Mbb^\pm=(\Mc^\pm,e|_{\Mc^\pm})$ where $\Mc^\pm=I^\pm_{\Mbb}(\Sigma^\pm)$. 
The obvious Cauchy morphisms induced by subset inclusions
\begin{equation}\label{eq:chain}
\Mbb \xlongleftarrow{\iota^-}{\Mbb^-} \xlongrightarrow{\tilde{\iota}^-}{\tilde{\Mbb}} \xlongleftarrow{\tilde{\iota}^+}  \act{\Lambda
}{\Mbb}^+  \xlongrightarrow{\iota^+}{\act{\Lambda
}{\Mbb}}
\end{equation}
(see Figure~\ref{fig:rot}) 
induce an isomorphism $\zeta(\tilde{\Lambda}):\Af(\Mbb)\to\Af(\act{\Lambda}{\Mbb})$,
\begin{equation}\label{eq:zeta1}
\zeta(\tilde{\Lambda}) = 
\Af(\iota^+)
\Af(\tilde{\iota}^+)^{-1}
\Af(\tilde{\iota}^-)
\Af(\iota^-)^{-1}
\end{equation}
by the timeslice property.
We will describe $\zeta(\tilde{\Lambda})$ as being formed by `chasing the arrows' in \eqref{eq:chain} from $\Mbb$ to $\act{\Lambda}{\Mbb}$. 
Note that if $\Af=\Bf\circ\FfL$ then, because $\FfL(\Mbb)=\FfL(\tilde{\Mbb})$, we have
$\iota^\pm=\tilde{\iota}^\pm$ (recall that each of these morphisms has
an inclusion as its underlying map) and hence $\zeta(\tilde{\Lambda})=\id_{\Bf(\Mbb)}$
for any $\tilde{\Lambda}$. 

We now show that the construction of $\zeta(\tilde{\Lambda})$ is independent of the choice
of $\Sigma^\pm$ and depends only on the homotopy class of $\tilde{\Lambda}$.
Starting with the Cauchy surfaces, note that whenever $\Sigma$ and $\Sigma'$
are smooth spacelike Cauchy surfaces, there is a smooth spacelike Cauchy surface
$\Sigma''$ in their common future (or past).
Hence it is enough to show that $\zeta(\tilde{\Lambda})$ is 
also obtained if smooth spacelike Cauchy surfaces $\hat{\Sigma}^\pm
\subset I_\Mbb^\pm(\Sigma^\pm)$ are used in place of $\Sigma^\pm$,
leaving $\tilde{\Lambda}$ unchanged.
Defining $\hat{\Mbb}{}^\pm$ by
analogy with $\Mbb^\pm$, the Cauchy morphisms of $\hat{\Mbb}{}^-$
into $\Mbb$ and $\tilde{\Mbb}$ factor via the Cauchy morphism
$j^-:\hat{\Mbb}{}^-\to\Mbb^-$, i.e., $\hat{\iota}^-=\iota^-\circ j^-$, 
and $\tilde{\hat{\iota}}^-=\tilde{\iota}^-\circ j^-$.  
Thus 
\begin{equation}
\Af(\tilde{\hat{\iota}}^-)
\Af(\hat{\iota}^-)^{-1}
=\Af(\tilde{\iota}^-)\Af(j^-)\Af(j^-)^{-1}\Af(\iota^-)^{-1}
=\Af(\tilde{\iota}^-)
\Af(\iota^-)^{-1};
\end{equation} 
a similar argument applies to $\Af(\iota^+)
\Af(\tilde{\iota}^+)^{-1}$ and establishes the required independence. 
Similarly, the isomorphism $\zeta(\tilde{\Lambda})$ is
also unchanged if we replace $\Mbb^\pm$ by causally convex subsets thereof
that contain Cauchy surfaces of $\Mbb$.   
\begin{figure}
\begin{center}
\begin{tikzpicture}[scale=0.7]
\definecolor{Green}{rgb}{0,.80,.20}
\definecolor{Gold}{rgb}{.93,.82,.24}
\definecolor{Orange}{rgb}{1,0.5,0}
\draw[fill=lightgray] (-6,0) -- ++(2,0) -- ++(0,3) -- ++(-2,0) -- cycle; 
\draw[fill=Gold] (-6,0.5) -- ++(2,0) -- ++(0,0.5) -- ++(-2,0) -- cycle; 

\draw[fill=Green] (6,0) -- ++(2,0) -- ++(0,3) -- ++(-2,0) -- cycle;
\draw[fill=Gold] (6,2) -- ++(2,0) -- ++(0,0.5) -- ++(-2,0) -- cycle; 

\draw[fill=lightgray] (0,0) -- ++(2,0) -- ++(0,1) -- ++(-2,0) -- cycle;
\draw[top color=Green,bottom color=lightgray] (0,1) -- ++(2,0) -- ++(0,1) -- ++(-2,0) -- cycle;
\draw[fill=Green] (0,2) -- ++(2,0) -- ++(0,1) -- ++(-2,0) -- cycle;
\draw[fill=Gold] (0,2) -- ++(2,0) -- ++(0,0.5) -- ++(-2,0) -- cycle;
\draw[fill=Gold] (0,0.5) -- ++(2,0) -- ++(0,0.5) -- ++(-2,0) -- cycle;

\draw[fill=Gold] (3,2) -- ++(2,0) -- ++(0,0.5) -- ++(-2,0) -- cycle;
\draw[fill=Gold] (-3,0.5) -- ++(2,0) -- ++(0,0.5) -- ++(-2,0) -- cycle;

\draw[color=blue,line width=2pt,->] (2.75,2.25) -- (2.25,2.25); 
\draw[color=blue,line width=2pt,->] (-0.75,0.75) -- (-0.25,0.75); 
\draw[color=blue,line width=2pt,->] (5.25,2.25) -- (5.75,2.25); 
\draw[color=blue,line width=2pt,->] (-3.25,0.75) -- (-3.75,0.75); 
\node[anchor=north] at (-5,0) {$\Mbb$};
\node[anchor=north] at (1,0) {$\tilde{\Mbb}$};
\node[anchor=north] at (7,0) {$\act{\Lambda}{\Mbb}$};
\node[anchor=north] at (4,0) {$\act{\Lambda}{\Mbb}^+$};
\node[anchor=north] at (-2,0) {$\Mbb^-$};
\end{tikzpicture}
\end{center}
\caption{Diagram of spacetimes involved in constructing  $\zeta(\tilde{\Lambda})$.}
\label{fig:rot}
\end{figure}
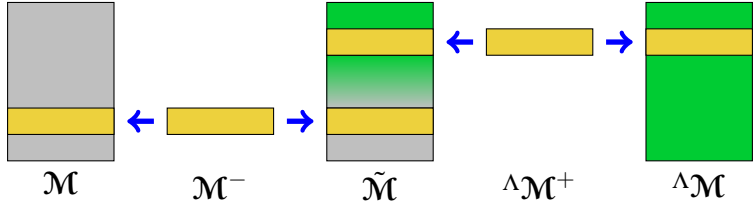

Next, let $\hat{\Lambda}\in C^\infty(\Mc;\Lc_0)$ obey the same conditions as
$\tilde{\Lambda}$ (relative to a common choice of Cauchy surfaces $\Sigma^\pm$ without loss of generality), thus also having homotopy class $S$. Then 
we have the following diagram of Cauchy morphisms
\begin{equation}\label{eq:Cauchy_diamond}
\begin{tikzpicture}[baseline=0 em, description/.style={fill=white,inner sep=2pt}]
\matrix (m) [ampersand replacement=\&,matrix of math nodes, row sep=-0.1em,
column sep=4em, text height=1.5ex, text depth=0.25ex]
{        \&                \& \tilde{\Mbb}   \&                  \&  \\
  \Mbb \& \Mbb^- \&      \& \act{\Lambda}{\Mbb}^+ \& \act{\Lambda}{\Mbb} \\
        \&                \& \hat{\Mbb}   \&                  \&  \\ };
\path[->]
(m-2-2) edge node[above] {$ \iota^- $} (m-2-1)
           edge node[above,sloped] {$ \tilde{\iota}^- $} (m-1-3)
           edge node[below,sloped] {$ \hat{\iota}^- $} (m-3-3)
(m-2-4) edge node[above] {$ \iota^+ $} (m-2-5)
           edge node[above,sloped] {$ \tilde{\iota}^+ $} (m-1-3)
           edge node[below,sloped] {$ \hat{\iota}^+ $} (m-3-3);
\end{tikzpicture}
\end{equation}
and the isomorphism $\zeta(\hat{\Lambda}):\Af(\Mbb)\to\Af(\act{\Lambda}{\Mbb})$ is
formed by chasing arrows along the lower branch from $\Mbb$ to $\act{\Lambda}{\Mbb}$. 
Now, $\hat{\Lambda}\tilde{\Lambda}^{-1}$ acts trivially outside
a timelike-compact set and is by assumption null-homotopic relative to the
complement of the time-compact subset of $\Mc$ bounded by $\Sigma^\pm$. 
Dynamical local Lorentz invariance then implies that 
$\rce_{\tilde{\Mbb}}[\hat{\Lambda}\tilde{\Lambda}^{-1}]$ is trivial, so
the diamond in \eqref{eq:Cauchy_diamond} commutes and  $\zeta(\tilde{\Lambda})=\zeta(\hat{\Lambda})$. As the isomorphism
depends only on the homotopy class $S$, it will henceforth be denoted
$\zeta_\Mbb(S)$.  
 
\paragraph{(b) Naturality of $\zeta(S)$} Let $\psi:\Mbb\to\Nbb$ be arbitrary. We must show that 
\begin{equation}\label{eq:zeta_nat}
\zeta_{\Nbb}(S)\circ\Af(\psi) = \Af(\act{S}{\psi})\circ\zeta_\Mbb(S) 
\end{equation}
holds for the isomorphisms defined in part (a). The obstruction to a straightforward proof of \eqref{eq:zeta_nat} is that the interpolating spacetime $\tilde{\Mbb}$ used to construct
$\zeta_\Mbb(S)$ (see Fig.~\ref{fig:rot}) might not embed in
an interpolating spacetime for the construction of $\zeta_{\Nbb}(S)$. Indeed, the function $\tilde{\Lambda}\in C^\infty(\Mbb;\Lc_0)$ might not be the pull-back of a function in $C^\infty(\Nbb;\Lc_0)$ -- for example, $\psi(\Mbb)$ might have boundary points to which the push-forward $\psi_*\tilde{\Lambda}$ cannot be extended continuously. We circumvent this
problem using an argument based on additivity. 

A further definition is required: Given open subsets $R^\pm$ of a spacetime $\Mbb\in\FLoc$, $\mu\in C^\infty(\Mbb,\arxivbreak [0,1])$ is a \emph{temporal mollifier} for the ordered pair $(R^-,R^+)$ if 
there exist smooth spacelike Cauchy surfaces $\Sigma^\pm$ for $\Mbb$
so that $R^\pm\subset I^\pm_\Mbb(\Sigma^\pm)$, and $\mu$ vanishes identically on $J^-_\Mbb(\Sigma^-)$ (hence on $R^-$) and is identically unity on $J_\Mbb^+(\Sigma^+)$ (hence on $R^+$). Temporal mollifiers exist for $(R^-,R^+)$ if and only
$R^{+\smash{/}-}$ lie to the future/past of a smooth spacelike Cauchy surface.

The proof of \eqref{eq:zeta_nat} falls into two parts: first, in Lemma~\ref{lem:local_zeta} 
we show that it holds on subspacetimes of $\Mbb$ provided suitable temporal mollifiers
can be found; second, in Lemma~\ref{lem:natural_zeta}, we show how such mollifiers
may be constructed on a sufficiently large class of subspacetimes to establish naturality if
$\Af$ is additive. 

\begin{lemma}\label{lem:local_zeta} Fix
$S\in\Sc$ and also a smooth path $\sigma:[0,1]\to\Sc$
with $\sigma(0)=\II$, $\sigma(1)=S$.
 (i) Let $\Mbb\in\FLoc$ and suppose morphisms 
$\rho^\pm:\Rbb^\pm\to\Mbb$ and $\upsilon:\Ubb\to\Mbb$ are given with image regions
obeying $\Im\rho^\pm\subset\Im\upsilon\subset D_\Mbb(\Im\rho^+)$ (see Fig.~\ref{fig:trio}(a)).
If $\mu$ is a temporal mollifier for $(\Im\rho^-,\Im\rho^+)$ in $\Mbb$, then 
\begin{equation}\label{eq:local_zeta}
\zeta_\Mbb(S)\Af(\rho^-) = \Af(\act{S}{\rho}^+)\Af(\kappa^+)^{-1}\Af(\kappa^-),
\end{equation}
where the morphisms $\Rbb^-\xlongrightarrow{\kappa^-}\act{\upsilon^*(\sigma\circ \mu)}{\Ubb}\xlongleftarrow{\kappa^+}\act{S}{\Rbb}^+$, of which $\kappa^+$ is Cauchy,
are uniquely determined by the requirement that 
$\FfL(\upsilon)\circ\FfL(\kappa^\pm)=\FfL(\rho^\pm)$.  

(ii) If, additionally, $\psi:\Mbb\to\Nbb$, suppose that a temporal mollifier $\nu$
exists for $(\Im\psi\circ\rho^-,\Im\psi\circ\rho^+)$ in $\Nbb$ that obeys $\psi^*\nu=\mu$ on $\upsilon(\Ubb)$. 
Then
\begin{equation}\label{eq:local_zeta_natural}
\zeta_{\Nbb}(S)\Af(\psi)\Af(\rho^-)=\Af(\act{S}{\psi})\zeta_\Mbb(S)\Af(\rho^-).
\end{equation}
\end{lemma}

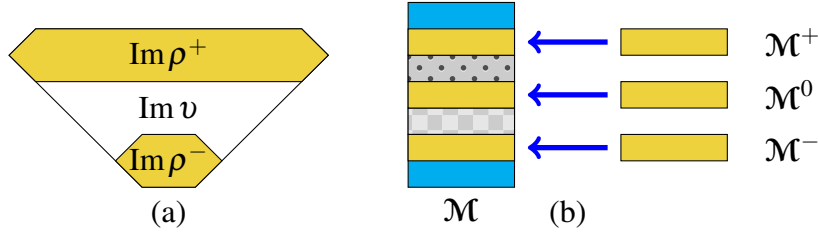
\begin{figure}
\begin{center}
\begin{tikzpicture}[scale=0.7]
\definecolor{Green}{rgb}{0,.80,.20}
\definecolor{Gold}{rgb}{.93,.82,.24}
\definecolor{Orange}{rgb}{1,0.5,0} 
\draw[fill=cyan] (-1,0) -- ++(2,0) -- ++(0,3.5) -- ++(-2,0) -- cycle; 
\draw[pattern=crosshatch dots gray] (-1,2) -- ++(2,0) -- ++(0,0.5) -- ++(-2,0) -- cycle;
\draw[pattern=checkerboard light gray] (-1,1) -- ++(2,0) -- ++(0,0.5) -- ++(-2,0) -- cycle;
\draw[fill=Gold] (-1,2.5) -- ++(2,0) -- ++(0,0.5) -- ++(-2,0) -- cycle;
\draw[fill=Gold] (-1,1.5) -- ++(2,0) -- ++(0,0.5) -- ++(-2,0) -- cycle;
\draw[fill=Gold] (-1,0.5) -- ++(2,0) -- ++(0,0.5) -- ++(-2,0) -- cycle;
\draw[fill=Gold] (3,2.5) -- ++(2,0) -- ++(0,0.5) -- ++(-2,0) -- cycle;
\draw[fill=Gold] (3,0.5) -- ++(2,0) -- ++(0,0.5) -- ++(-2,0) -- cycle;
\draw[fill=Gold] (3,1.5) -- ++(2,0) -- ++(0,0.5) -- ++(-2,0) -- cycle;
\node[anchor=west] at (5.5,2.75) {$\Mbb^+$};
\node[anchor=west] at (5.5,1.75) {$\Mbb^0$};
\node[anchor=west] at (5.5,0.75) {$\Mbb^-$};
\draw[color=blue,line width=2pt,->] (2.75,2.75) -- ++(-1.5,0);
\draw[color=blue,line width=2pt,->] (2.75,1.75) -- ++(-1.5,0);
\draw[color=blue,line width=2pt,->] (2.75,0.75) -- ++(-1.5,0);
\node[anchor=north] at (0,0) {$\Mbb$}; 
\draw (-6,0) -- ++ (1,0) -- ++(2.5,2.5) -- ++(-0.5,0.5) -- ++(-5,0) -- ++(-0.5,-0.5) -- cycle;
\draw[fill=Gold] (-6,0) -- ++ (1,0) -- ++(0.5,0.5) -- ++(-0.5,0.5) -- ++(-1,0) -- ++(-0.5,-0.5) -- cycle;
\draw[fill=Gold] (-8,3) -- ++ (5,0) -- ++(0.5,-0.5) -- ++(-0.5,-0.5) -- ++(-5,0) -- ++(-0.5,0.5) -- cycle;
\node at (-5.5,0.5) {$\Im \rho^-$};
\node at (-5.5,2.5) {$\Im \rho^+$};
\node at (-5.5,1.5) {$\Im \upsilon$};
\node at (-5.5,-0.5) {(a)};
\node at (2,-0.5) {(b)};
\end{tikzpicture}
\end{center}
\caption{(a) The arrangement of image regions $\Im\rho^\pm\subset\Im\upsilon\subset D_\Mbb(\Im\rho^+)$ used in Lemma~\ref{lem:local_zeta}. (b) The spacetimes used to compute the cocycle. The map $\tilde{\Lambda}_T$ (resp. $\tilde{\Lambda}_S$) is locally constant outside the chequered (resp., dotted) region.}
\label{fig:trio}
\end{figure}

\begin{proof} 
(i) Select smooth spacelike Cauchy surfaces $\Sigma^\pm$ in $\Mbb$ so that
$\mu$ vanishes to the past of $\Sigma^-$ and is identically unity on 
the future of $\Sigma^+$, arranging also that $\Im\rho^\pm\subset
I_\Mbb^\pm(\Sigma^\pm)$. Using these Cauchy surfaces to define
$\Mbb^\pm$ as in part (a), and building the interpolating spacetime $\tilde{\Mbb}$
using $\tilde{\Lambda}=\pi(\tilde{S})$ where $\tilde{S}=\sigma\circ\mu$, 
we construct $\zeta_\Mbb(S)$ by chasing
arrows from left to right along the top line of  
\[
\includegraphics[]{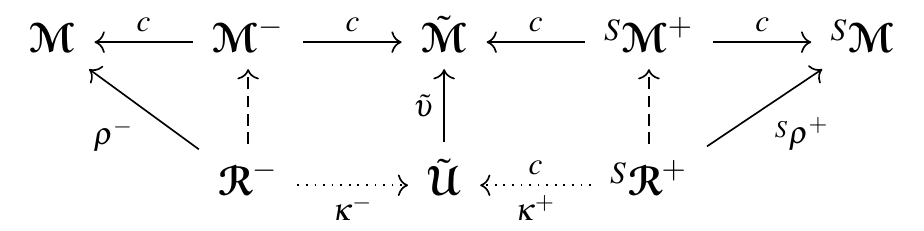}
\]
in which $\tilde{\Ubb}=\act{\upsilon^*\tilde{S}}{\Ubb}$ and $\tilde{\upsilon}=\act{\upsilon^*\tilde{S}}{\upsilon}$. We now  establish the existence of the dashed and dotted morphisms and show that the diagram  commutes, from which~\eqref{eq:local_zeta} follows by functoriality and the timeslice
property.
As $\Im\rho^\pm\subset I_\Mbb^\pm(\Sigma^\pm)$, there 
are unique dashed morphisms as shown, making the two triangles commute, and inducing morphisms from $\Rbb^-$ and
$\act{S}{\Rbb}^+$ to $\tilde{\Mbb}$ via $\Mbb^-$ and $\act{S}{\Mbb}^+$;
these morphisms have the same underlying functions as $\rho^\pm$. The conditions
on $\Im\rho^\pm$ and $\Im\upsilon$ entail that there are (unique)
morphisms $\kappa^\pm$ making the squares commute and with $\kappa^+$ Cauchy, with
underlying functions obeying $\upsilon\circ\kappa^\pm=\rho^\pm$; more formally we may write 
$\FfL(\upsilon)\circ\FfL(\kappa^\pm)=\FfL(\rho^\pm)$, and this relation
determines $\kappa^\pm$ uniquely because their codomain $\tilde{\Ubb}$ is fixed.   
Part (i) is thus proved.  

(ii) We apply part (i) to $\psi\circ\rho^\pm:\Rbb^\pm\to\Nbb$, and $\psi\circ\upsilon:\Ubb\to\Nbb$
using $\pi(\sigma\circ\nu)$ to build an interpolating spacetime for the construction of
$\zeta_\Nbb(S)$. Note that $(\psi\circ\upsilon)^*\nu=\upsilon^*\psi^*\nu=\upsilon^*\mu$, so 
\[
\act{(\psi\circ\upsilon)^*(\sigma\circ \nu)}{\Ubb}=\act{\upsilon^*(\sigma\circ \mu)}{\Ubb}=\tilde{\Ubb}.
\] 
Therefore one has $\zeta_\Nbb(S)\Af(\psi\circ\rho^-) = 
\Af(\act{S}{(\psi\circ\rho}^+))\Af(\kappa^+)^{-1}\Af(\kappa^-)$
with the same morphisms $\kappa^\pm$ as in the original application of part (i), because those morphisms obviously satisfy the characterising equation $\FfL(\psi\circ\upsilon)\circ\FfL(\kappa^\pm)=\FfL(\psi\circ\rho^\pm)$ and have the same codomain $\tilde{\Ubb}$. Combining this with \eqref{eq:local_zeta} gives
\[
\zeta_{\Nbb}(S)\Af(\psi)\Af(\rho^-)=  \Af(\act{S}{\psi})\Af(\act{S}{\rho}^+)
\Af(\kappa^+)^{-1}\Af(\kappa^-) = \Af(\act{S}{\psi})\zeta_\Mbb(S)\Af(\rho^-).\qquad\cmpqed
\] 
\end{proof} 
The above circumstances can be achieved
for sufficiently many $\rho^-$ to form a jointly epic set of morphisms 
$\Af(\rho^-)$.
\begin{lemma} \label{lem:natural_zeta}
Suppose $\psi:\Mbb\to\Nbb$ in $\FLoc$ and let $S\in\Sc$. For any truncated
multi-diamond $D$ of $\Mbb$, we have
\begin{equation}\label{eq:local_zeta_natural2}
\zeta_\Nbb(S)\Af(\psi)\Af(\iota_{\Mbb;D}) = \Af(\act{S}{\psi})\zeta_\Mbb(S)\Af(\iota_{\Mbb;D}).
\end{equation}
\end{lemma}
The additivity assumption on $\Af$ entails that the $\Af(\iota_{\Mbb;D})$ are jointly epic as $D$ runs over the truncated multi-diamonds. Therefore 
one has $\zeta_\Nb(S)\Af(\psi)=\Af(\act{S}{\psi})\zeta_\Mbb(S)$, 
and as $\psi:\Mbb\to\Nbb$ was arbitrary naturality of $\zeta(S)$ is established. 
\begin{proof}[\ifnonCMP Proof \fi of Lemma~\ref{lem:natural_zeta}]
Let $D$ be based on a Cauchy multi-ball $B^-\subset\Sigma$, where $\Sigma$ is a smooth spacelike Cauchy surface. 
By~\cite[Thm 1.2]{Bernal:2005qf}, we may find a Cauchy temporal
function foliating $\Mc$ as $\RR\times\Sigma$ so that $D$ has base $\{0\}\times B^-$ and the $e$-metric splits orthogonally as $\beta \oplus -h_t$, where $h_t$ is a smooth
Riemannian metric on $\Sigma$ for each $t\in\RR$ and $\beta\in C^\infty(\Mc)$ is nonnegative.   
The significance of this splitting is that each $\{t\}\times\Sigma$ is a smooth
spacelike Cauchy surface and all curves $t\mapsto(t,\sigma)$ for fixed $\sigma\in\Sigma$ 
are timelike. This facilitates bounds on Cauchy developments, e.g., $D_\Mbb(\{t\}\times B)\subset
\RR\times B$. 

The form of the $e$-metric allows us to choose another Cauchy multi-ball $\{0\}\times B^+$ containing the closure of $\{0\}\times B^-$ and $\epsilon>0$ such that $\{0\}\times B^-\subset D_\Mbb(\{t\}\times B^+)$ for all $0<t<\epsilon$
(cf.~\cite[Lem.~2.5]{Few_split:2015}). Choosing $t^+\in(0,\epsilon)$ and setting $t^-=0$, we define truncated multi-diamonds
\[
R^\pm =\{(t,\sigma)\in  D_\Mbb(\{t^\pm\}\times B^\pm):~|t-t^\pm|<t^+/3\}  
\]
based on the Cauchy multi-balls $\{t^\pm\}\times B^\pm$. Setting 
\[
U=\{(t,\sigma)\in D_\Mbb(R^+): -t^+/3<t<4t^+/3\},
\]
it is evident that $R^\pm\subset U\subset D_\Mbb(R^+)$. We may choose a temporal mollifier $\mu$ for $(R^-,R^+)$ so that $\mu$ vanishes on $(-\infty,4t^+/9]\times\Sigma$ and $\mu$ is unity on $[5t^+/9,\infty)\times\Sigma$. 

Supposing that $\psi:\Mbb\to\Nbb$, we now construct a temporal mollifier $\nu$ for $(\psi(R^-),\psi(R^+))$ so that $\psi^*\nu$ and $\mu$ agree on 
$L$. Choose  Cauchy surfaces $\Sigma^-$ (resp., $\Sigma^+$) in $\Nbb$ containing the Cauchy multi-ball  $\psi(\{4t^+/9\}\times B^+)$ (resp., $\psi(\{5t^+/9\}\times B^+)$) 
and so that $\Sigma^\pm\subset I_\Nbb^\pm(\Sigma^\mp)$.\footnote{As $\psi(\{4t^+/9\}\times B^+)$ is a Cauchy multi-ball it lies in a smooth spacelike Cauchy surface
$\Sigma^-$ of $\Nbb$. Similarly, there exists $\Sigma^+$ in the (globally hyperbolic region) $I_\Nbb^+(\Sigma^-)$ containing $\psi(\{5t^+/9\}\times B^+)$. As $\Sigma^\pm$ are Cauchy surfaces, one also has $\Sigma^-\subset I_\Nbb^-(\Sigma^+)$.} Owing to the split form of the $e$-metric, $R^\pm\subset I^\pm\times B^\pm$, where $I^\pm=\{t\in\RR:~|t-t^\pm|<t^+/3\}$. Accordingly, 
$R^+\subset (2t^+/3,4t^+/3)\times B^+\subset I^+_\Mbb(\{5t^+/9\}\times B^+)$ and hence $\psi(R^+)\subset I^+_\Nbb(\Sigma^+)$;
similarly $R^-\subset (-t^+/3,t^+/3)\times B^-\subset I^-_\Mbb(\{4t^+/9\}\times B^+)$ so $\psi(R^-)\subset I^-_\Nbb(\Sigma^-)$. Let $F$ be the closed set formed as the union
of $J_\Nbb^\pm(\Sigma^\pm)$ and the closure of $\psi(U)$. Due to the properties of 
$\mu$ and $\Sigma^\pm$, we may choose a smooth function $\nu$ on $F$ that vanishes
on $J_\Nbb^-(\Sigma^-)$, is identical to unity on $J_\Nbb^+(\Sigma^+)$, and 
agrees with $\mu\circ\psi^{-1}$ on $\overline{\psi(U)}$.
Every point $p\in F$ has a neighbourhood in which $\nu$ can be extended to a smooth function taking values in $[0,1]$ -- for  $p\in J^\pm_\Nbb(\Sigma^\pm)$ this is
obvious, while for $p$ in the closure of $\psi(U)$ one may use $\mu\circ\psi^{-1}$. 
By the smooth Tietze extension theorem (a partition of unity argument) one may obtain an 
extension of $\nu$ in $C^\infty(\Nbb,[0,1])$. In particular, $\nu$ is a temporal mollifier for $(\psi(R^-),\psi(R^+))$ and $\psi^*\nu$ agrees with $\mu$ on 
$U$. 

Letting $\Rbb^\pm=\Mbb|_{R^\pm}$, $\rho^\pm=\iota_{\Mbb;R^\pm}$, $\Ubb=\Mbb|_U$, and $\upsilon=\iota_{\Mbb;U}$,  parts (i) and (ii) of Lemma~\ref{lem:local_zeta} apply and give \eqref{eq:local_zeta_natural}. By the timeslice property, 
$\Af(\iota_{\Mbb;D})=\Af(\iota_{\Mbb;R^-})\circ \vartheta$ for
some isomorphism $\vartheta$, because $D$ and $R^-$ are
truncated multi-diamonds with a common base (there
are Cauchy morphisms from $\Mbb|_{D\cap R^-}$ to each
of $\Mbb|_D$ and $\Mbb_{R^-}$). Therefore
\eqref{eq:local_zeta_natural2} holds. \cmpqed
\end{proof}

\paragraph{(c) Computation of the $2$-cocycle} The construction of 
$S\mapsto\zeta(S)$ shows that $\Af$ is $\Sc$-covariant; we now show that
the corresponding cocycle $(\xi,\phi)$ is trivial. Starting with $\phi$,
let $\alpha\in\Aut(\Af)$ and consider the diagram
\[
\includegraphics[]{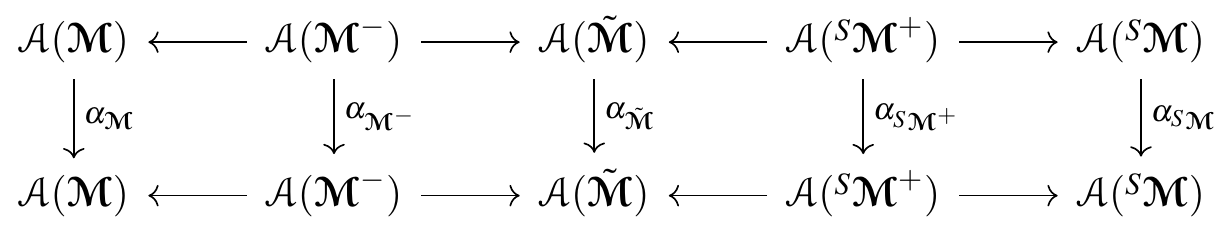}
\]
in which unlabelled arrows are isomorphisms arising as images of
Cauchy morphisms in \eqref{eq:chain} and each arrow on the bottom row is identical to the one vertically above it.  Each square commutes by naturality of $\alpha$
and one has
$\zeta(S)_\Mbb\alpha_\Mbb = \alpha_{\act{S}{\Mbb}}\zeta(S)_\Mbb$
by definition of $\zeta(S)$. Thus $\phi(S)(\alpha)=\alpha$, for all $S$ and $\alpha$.

It remains to prove that 
\begin{equation}\label{eq:triv_zeta}
\zeta_\Mbb(ST)= \zeta_{\act{T}{\Mbb}}(S) \zeta_\Mbb(T) \qquad(S,T\in\Sc,~\Mbb\in\FLoc).
\end{equation}
Fix $\Mbb=(\Mc,e)\in\FLoc$ and choose a Cauchy temporal function $\tau\in C^\infty(\Mc,\RR)$  
for $\Mbb$ -- i.e., $\nabla\tau$ is everywhere $e$-timelike and future-pointing, and the level sets of $\tau$ are smooth spacelike Cauchy surfaces. Also choose
disjoint open bounded intervals $I^-,I^0,I^+$ of $\RR$ such that $0\in I^0$ and $I^\pm\subset\RR^\pm$, and define submanifolds $\Mc^{-\smash{/}0\smash{/}+}=\tau^{-1}(I^{-\smash{/}0\smash{/}+})$. 
Finally, choose $\tilde{S},\tilde{T}\in C^\infty(\Mc,\Sc)$  
so that $\tilde{S}\equiv 1$ on $J^-_\Mbb(\Mc^0)$ and $\tilde{S}\equiv S$ in $J^+_\Mbb(\Mc^+)$
while $\tilde{T}\equiv 1$ on $J^-_\Mbb(\Mc^-)$ and $\tilde{T}\equiv T$ on $J^+_\Mbb(\Mc^0)$.
Then $\tilde{\Lambda}_S=\pi(S)$ and $\tilde{\Lambda}_T=\pi(T)$ have the homotopy types of $S$ and $T$ relative to  $J_\Mbb^-(\Mc^0)\cup J^+_\Mbb(\Mc^+)$ and 
to $J_\Mbb^-(\Mc^-)\cup J^+_\Mbb(\Mc^0)$ respectively. Evidently 
$\tilde{\Lambda}_S\tilde{\Lambda}_T\in C^\infty(\Mc,\Lc_0)$ takes the constant values $1$ on $J_\Mb^-(\Mc^-)$, $\pi(T)$ on $\Mc^0$ and  
$\pi(ST)$ on $J_\Mb^+(\Mc^+)$, and has the homotopy type of $ST$ relative to 
$J_\Mbb^-(\Mc^-)\cup J^+_\Mbb(\Mc^+)$.
Given these choices, we define various spacetimes: $\Mbb^{-\smash{/}0\smash{/}+}=\Mbb|_{\Mc^{-\smash{/}0\smash{/}+}}$
are slabs of $\Mbb$ sandwiching the regions where $\tilde{\Lambda}_S$ and $\tilde{\Lambda}_T$
can vary (see Fig.~\ref{fig:trio}), while 
\[
\tilde{\Mbb}_T=(\Mc,\tilde{\Lambda}_T e), \qquad\tilde{\Mbb}_S=(\Mc,\tilde{\Lambda}_S \pi(T) e)
\qquad\text{and}\quad
\tilde{\Mbb}_{ST}=(\Mc,\tilde{\Lambda}_S\tilde{\Lambda}_T e)
\]
are interpolating spacetimes used in the constructions of $\zeta_\Mbb(T)$, $\zeta_{\act{T}{\Mbb}}(S)$
and $\zeta_\Mbb(ST)$ respectively. 
Now consider the following diagram of Cauchy morphisms:
\begin{equation}\label{eq:diag1}
\begin{tikzpicture}[baseline=0 em, description/.style={fill=white,inner sep=2pt}]
\matrix (m) [ampersand replacement=\&,matrix of math nodes, row sep=0.8em,
column sep=1.5em, text height=1.5ex, text depth=0.25ex]
{ \Mbb       \& \Mbb^-  \& \tilde{\Mbb}_T       \&                              \&       \\
                 \&  \tilde{\Mbb}_{ST}               \&         \& \act{T}{\Mbb_0} \& \act{T}{\Mbb} \\
\act{ST}{\Mbb} \& \act{ST}{\Mbb}^+  \& \tilde{\Mbb}_{S}   \&                              \&  \\ };
\path[->]
(m-1-2)  edge (m-1-1)
             edge (m-1-3)
             edge (m-2-2)
(m-3-2)  edge (m-3-1)
             edge (m-3-3)
             edge (m-2-2)
(m-2-4) edge  (m-2-5)
            edge[dashed]  (m-2-2)
            edge  (m-1-3)
            edge  (m-3-3);
\end{tikzpicture}
\end{equation}
all of which are induced by inclusion morphisms. (The dashed morphism is well-defined because $\tilde{\Lambda}_S\tilde{\Lambda}_T$ takes the constant value $\pi(T)$ in $\Mc^0$.)
The isomorphism $\zeta_\Mbb(T)$
is obtained by chasing the images under $\Af$ of the arrows, left to right, on the upper line from $\Mbb$ to $\act{T}{\Mbb}$; 
$\zeta_{\act{T}{\Mbb}}(S)$ is obtained by chasing the arrows on the lower line, right to left, from $\act{T}{\Mbb}$ to $\act{ST}{\Mbb}$, while $\zeta_\Mbb(ST)$ is obtained by chasing 
from $\Mbb$ to $\act{ST}{\Mbb}$ along the shortest route. One sees that the identity \eqref{eq:triv_zeta} can be proved by focussing on the central portion of diagram~\eqref{eq:diag1}
(deleting the external legs) and establishing that the isomorphism from $\Af(\Mbb^-)$ to $\Af(\act{ST}{\Mbb}^+)$ induced by chasing via $\act{T}{\Mbb}^0$ is equal to that obtained by chasing vertically downwards. Using the dashed arrow the task splits into two problems, with diagrams
\[
\begin{tikzpicture}[baseline=0 em, description/.style={fill=white,inner sep=2pt}]
\matrix (m) [ampersand replacement=\&,matrix of math nodes, row sep=0.1em,
column sep=2em, text height=1.5ex, text depth=0.25ex]
{      \Mbb^-  \& \& \tilde{\Mbb}_T                                             \\
\& \tilde{\Mbb} \& \\
        \tilde{\Mbb}_{ST}               \&         \& \act{T}{\Mbb}^0  \\ };
\path[->]
(m-1-1)  edge (m-1-3)
             edge (m-3-1) 
(m-3-3) edge  (m-1-3)
            edge[dashed]  (m-3-1)
(m-1-1) edge[dotted]  (m-2-2) 
(m-2-2) edge[dotted]  (m-1-3)
(m-2-2) edge[dotted]  (m-3-1)
(m-3-3) edge[dotted]  (m-2-2);       
\end{tikzpicture}
\quad\text{and}\quad
\begin{tikzpicture}[baseline=0 em, description/.style={fill=white,inner sep=2pt}]
\matrix (m) [ampersand replacement=\&,matrix of math nodes, row sep=0.1em,
column sep=2em, text height=1.5ex, text depth=0.25ex]
{ 
    \tilde{\Mbb}_{ST}           \&   \&           \act{T}{\Mbb}^0    \\
   \& \hat{\Mbb} \& \\
  \act{ST}{\Mbb}^+  \& \& \tilde{\Mbb}_{S}                         \\ };
\path[->] 
(m-3-1)  edge (m-3-3)
             edge (m-1-1) 
(m-1-3) edge[dashed]  (m-1-1)
            edge  (m-3-3)
(m-2-2) edge[dotted]  (m-1-1) 
(m-1-3) edge[dotted]  (m-2-2)
(m-3-1) edge[dotted]  (m-2-2)
(m-2-2) edge[dotted]  (m-3-3); 
\end{tikzpicture}
\] 
where again we must show the equivalence of the chase around the right-hand portions 
to that obtained by passing vertically down from $\Mbb^-$ or $\tilde{\Mbb}_{ST}$ respectively.
In these diagrams, the solid and dashed morphisms are those in the previous diagram, 
$\tilde{\Mbb}=\tilde{\Mbb}_{T}|_{J^-_\Mbb(\Mc^0)}$,  $\hat{\Mbb}=\tilde{\Mbb}_{ST}|_{J^+_\Mbb(\Mc^0)}$, 
and the dotted morphisms are defined by inclusion maps. 
Every small triangle in these diagrams is a commuting triangle of Cauchy morphisms
induced by an inclusion. Taking images under $\Af$, every small triangle is a commuting triangle of isomorphisms and therefore the isomorphisms induced by chasing along the right-hand portions of the diagrams coincide with the isomorphism
induced by the left-hand vertical line. This concludes the proof of~\eqref{eq:triv_zeta}
and hence of Theorem~\ref{thm:CM}.
 
\subsection{Minkowski space}\label{sec:Mink}

Define $n$-dimensional Minkowski space to be the object $\Mbb_0=(\RR^n,(dX^\mu)_{\mu=0}^{n-1})$ of $\FLoc$, where $X^\mu:\RR^n\to\RR$ are the coordinate functions $X^\mu(x^0,\ldots,x^{n-1})=x^\mu$. The corresponding object $\Mb_0:=\FfL(\Mbb_0)$ of $\Loc$ has the restricted Poincar\'e group as its group of automorphisms. In $\FLoc$, however, Lorentz symmetry is broken by the choice of frame and $\Mbb_0$ only admits spacetime translations as automorphisms. Instead, Lorentz
transformations map between distinct objects: each 
$\Lambda\in \Lc_0$ induces an active Lorentz transformation $\psi_\Lambda:\RR^n\to\RR^n$
by matrix multiplication, $X^\mu \circ \Lambda =\Lambda^\mu_{\phantom{\mu}\nu}X^\nu$; as
$\psi_\Lambda^*dX^\mu=\Lambda^\mu_{\phantom{\mu}\nu}
dX^\nu$, $\psi_\Lambda$ induces an morphism $\psi_\Lambda: \Mbb_0\to \act{\Lambda^{-1}}{\Mbb_0}$ in $\FLoc$. Given a second $\Lambda'\in \Lc_0$, the
morphism $\act{\Lambda^{-1}}{\psi}_{\Lambda'}: \act{\Lambda^{-1}}{\Mbb}_0\to 
\act{(\Lambda'\Lambda)^{-1}}{\Mbb}_0$
has underlying map $\Lambda'$, and therefore $\act{\Lambda^{-1}}{\psi}_{\Lambda'}\circ\psi_\Lambda$ has underlying map $\Lambda'\Lambda$, giving the equality of morphisms
\begin{equation}\label{eq:psi_cp}
\psi_{\Lambda'\Lambda}=\act{\Lambda^{-1}}{\psi}_{\Lambda'}\circ\psi_\Lambda.
\end{equation}
As $\psi_\Lambda$ has inverse $\act{\Lambda^{-1}}{\psi}_{\Lambda^{-1}}$ it is therefore an isomorphism in $\FLoc$. Of course, $\FfL(\psi_\Lambda)$ is simply the 
Lorentz transformation $\Lambda$ as an automorphism of $\Mb_0$. To economize
on notation we also write $\psi_S$ for $\psi_{\pi(S)}$ if $S\in\Sc$.

Whereas a functor on $\Loc$ automatically represents Lorentz transformations
in the automorphism group of $\Af(\Mb_0)$, the same is not true of functors
on $\FLoc$. This is remedied precisely by $\Sc$-covariance: 
for each $S\in\Sc$, we may define 
\begin{equation}
\Xi(S)=\Af(\act{S}{\psi_S}) \circ\zeta(S)_{\Mbb_0}=
\zeta(S)_{\act{S^{-1}}{\Mbb_0}}\circ\Af(\psi_{S}),
\end{equation}
an automorphism of $\Af(\Mbb_0)$ with some important properties. 
First, it is clear that $\Xi(1)=\id_{\Af(\Mbb_0)}$ and more generally that, 
if $S\in\ker\pi$ covers the identity
Lorentz transformation, then $\Xi(S)=\zeta(S)_{\Mbb_0}$. For example,
in $n\ge 4$ spacetime dimensions, $\Xi(-1)=\zeta(-1)_{\Mbb_0}$.
Second, note that  
\begin{align}
\Xi(S')\Xi(S) &= \Af(\act{S'}{\psi_{S'}})\zeta(S')_{\Mbb_0}
 \Af(\act{S}{\psi_{S}})\zeta(S)_{\Mbb_0} \nonumber\\&
= \Af(\act{S'}{\psi_{S'}})
 \Af(\act{S}{\psi_S})\zeta(S')_{\act{S}{\Mbb_0}} \zeta(S)_{\Mbb_0} \nonumber \\ &
 = \Af(\act{S'S}{\psi_{S'S}})\zeta(S'S)_{\Mbb_0}  =\Xi(S'S),
\end{align}
where, in the penultimate step, we use the fact that $\zeta$ has
a trivial cocycle, and also the identity \eqref{eq:psi_cp}.
Third, the action on fields is 
\begin{align*}
\Xi(S)\Phi_{\Mbb_0}(f) &=  \Af(\act{S}{\psi_S})\zeta(S)_{\Mbb_0} 
\Phi_{\Mbb_0}(f)  
= \Af(\act{S}{\psi_S}) (S\ast\Phi)_{\act{S}{\Mbb_0}}( f)    \\
& = (S\ast\Phi)_{\Mbb_0}(\pi(S)_* f) 
\end{align*}
for all $f\in\CoinX{\RR^n}$ ($\Df$ is $\Sc$-covariant with a trivial implementation).  
 
Fourth, given a state $\omega_0$ on $\Af(\Mbb_0)$ that is invariant under these automorphisms, 
i.e., $\omega_0\circ\Xi(S)=\omega_0$ for all $S$, the corresponding GNS representation $(\HH_0,\DD_0,\pi_0,\Omega_0)$ will carry a unitary implementation of the $\Xi(S)$, so that
\[
\pi_0(\Xi(S)A)=U(S)\pi_0(A)U(S)^{-1}, \qquad U(S)\Omega_0=\Omega_0
\]
for all $S\in\Sc$, recovering the standard transformation laws
of fields in Minkowski QFT. The use of $\FLoc$ has distinguished two aspects of
the Lorentz transformation: the active transformation of points and algebra elements 
under $\Af(\psi_S)$, and the passive relabelling of field multiplets arising from $\Sc$-covariance.

\section{Conclusion}\label{sec:conclusion}

We have given a general analysis of $G$-covariance of locally
covariant theories in terms of nonabelian cohomology. Among the general features
uncovered are the existence of an associated canonical cohomology class, and the structure of field multiplets. As well as discussing rigid scale covariance,
we have established a no-go theorem on mixing of internal and Lorentz symmetries
analogous to the Coleman--Mandula theorem. Our approach here is completely new
and makes no use of Minkowski spacetime structures. This gives a new perspective
on results of this type and further demonstrates the utility of  
relative Cauchy evolution. 

Directions in which this work could be extended include the following.
First one could study smooth $G$-covariance using, e.g., the
smooth nonabelian cohomology of~\cite{Neeb:2007}. Topologies on $\Aut(\Af)$ and
$\Fld(\Af)$ can be given in terms of suitable state spaces~\cite[\S 3.2]{Fewster:gauge}. 
Second, the proof of Theorem~\ref{thm:CM} would apply to other rigid group actions that can be achieved by smooth deformation (e.g., the conformal group).
Finally, an obvious question is whether an analogue of the Haag--{\L}opusza\'{n}ski--Sohnius theorem~\cite{HaagLopSoh:1975} can be proved for theories on a suitable category of supermanifolds, perhaps using the enriched category methods of~\cite{HacHanSch:2016}.

\begin{acknowledgement}  It is a pleasure to thank Michael Bate and Stephen
Donkin for discussions on group cohomology and Klaus Fredenhagen, Markus Fr\"ob, Igor Khavkine, Nicola Pinamonti, 
Katarzyna Rejzner, Rainer Verch, for useful conversations and remarks on the main results.  
I am also grateful to the organisers and participants of the workshop
{\em Algebraic quantum field theory: its status and its future} held
at the Erwin Schr\"odinger Institute, Vienna, May 2014, at which 
some of these ideas were first presented and further developed. 
\end{acknowledgement} 


%

{\small
\begin{thebibliography}{10}\setlength{\itemsep}{-1.5mm}
	\providecommand{\url}[1]{{#1}}
	\providecommand{\urlprefix}{URL }
	\providecommand{\href}[2]{#2}
	\expandafter\ifx\csname urlstyle\endcsname\relax
	\providecommand{\doi}[1]{DOI~\discretionary{}{}{}#1}\else
	\providecommand{\doi}{DOI~\discretionary{}{}{}\begingroup
		\urlstyle{rm}\Url}\fi
	
	\bibitem{AzcaIzqu:1995}
	de~Azc{\'a}rraga, J.A., Izquierdo, J.M.:
	\href{http://dx.doi.org/10.1017/CBO9780511599897}{Lie groups, {L}ie algebras,
		cohomology and some applications in physics.} 
	\newblock Cambridge Monographs on Mathematical Physics. Cambridge University
	Press, Cambridge (1995)
	
	\bibitem{Bernal:2005qf}
	Bernal, A.N., S{\'{a}}nchez, M.: {Further results on the smoothability of
		Cauchy hypersurfaces and Cauchy time functions}.
	\newblock \href{http://dx.doi.org/10.1007/s11005-006-0091-5}{Lett. Math. Phys.
		\textbf{77}, 183--197 (2006)}
	
	\bibitem{BruDueFre2009}
	Brunetti, R., D{\"u}tsch, M., Fredenhagen, K.: Perturbative algebraic quantum
	field theory and the renormalization groups.
	\newblock \href{http://dx.doi.org/10.4310/ATMP.2009.v13.n5.a7}{Adv. Theor.
		Math. Phys. \textbf{13}, 1541--1599 (2009)}
	
	\bibitem{BrFrVe03}
	Brunetti, R., Fredenhagen, K., Verch, R.: The generally covariant locality
	principle: A new paradigm for local quantum physics.
	\newblock \href{http://dx.doi.org/10.1007/s00220-003-0815-7}{Commun. Math.
		Phys. \textbf{237}, 31--68 (2003)}
	
	\bibitem{Br&Ru05}
	Brunetti, R., Ruzzi, G.: Superselection sectors and general covariance. {I}.
	\newblock \href{http://dx.doi.org/10.1007/s00220-006-0147-5}{Commun. Math.
		Phys. \textbf{270}, 69--108 (2007)}
	
	\bibitem{ChiFre:2009}
	Chilian, B., Fredenhagen, K.: The time slice axiom in perturbative quantum
	field theory on globally hyperbolic spacetimes.
	\newblock \href{http://dx.doi.org/10.1007/s00220-008-0670-7}{Comm. Math. Phys.
		\textbf{287}, 513--522 (2009)}
	
	\bibitem{ColeMand:1967}
	{Coleman}, S., {Mandula}, J.: {All Possible Symmetries of the S Matrix}.
	\newblock \href{http://dx.doi.org/10.1103/PhysRev.159.1251}{Physical Review
		\textbf{159}, 1251--1256 (1967)}
	
	\bibitem{CostaMorisson:2016}
	Costa, R., Morrison, I.A.: On higher spin symmetries in de {S}itter {QFT}s.
	\newblock \href{http://dx.doi.org/10.1007/JHEP03(2016)056}{Journal of High
		Energy Physics \textbf{2016}, 1--20 (2016)}
	
	\bibitem{DabBro:2014}
	Dabrowski, Y., Brouder, C.: Functional properties of {H}\"ormander's space of
	distributions having a specified wavefront set.
	\newblock \href{http://dx.doi.org/10.1007/s00220-014-2156-0}{Comm. Math. Phys.
		\textbf{332}, 1345--1380 (2014)}
	
	\bibitem{DeWBre:1960}
	DeWitt, B.S., Brehme, R.W.: Radiation damping in a gravitational field.
	\newblock \href{http://dx.doi.org/10.1016/0003-4916(60)90030-0}{Ann. Physics
		\textbf{9}, 220--259 (1960)}
	
	\bibitem{EilMac:1947b}
	Eilenberg, S., MacLane, S.: Cohomology theory in abstract groups. {II}. {G}roup
	extensions with a non-{A}belian kernel.
	\newblock \href{http://dx.doi.org/10.2307/1969174}{Ann. of Math. (2)
		\textbf{48}, 326--341 (1947)}
	
	\bibitem{Ferguson_PhD}
	Ferguson, M.T.: Aspects of dynamical locality and locally covariant canonical
	quantization.
	\newblock Ph.D. thesis, University of York (2013).
	\newblock \urlprefix\url{http://etheses.whiterose.ac.uk/4529/}
	
	\bibitem{Few_spinstats}
	Fewster, C.J.: The spin--statistics connection in curved spacetimes.
	\newblock In preparation
	
	\bibitem{Fewster2007}
	Fewster, C.J.: Quantum energy inequalities and local covariance. {II}.
	{C}ategorical formulation.
	\newblock \href{http://dx.doi.org/10.1007/s10714-007-0494-3}{Gen. Relativity
		Gravitation \textbf{39}, 1855--1890 (2007)}
	
	\bibitem{Fewster:gauge}
	Fewster, C.J.: Endomorphisms and automorphisms of locally covariant quantum
	field theories.
	\newblock \href{http://dx.doi.org/10.1142/S0129055X13500086}{Rev. Math. Phys.
		\textbf{25}, {1350}{008}, 47 (2013)}
	
	\bibitem{Few_split:2015}
	Fewster, C.J.: The {S}plit {P}roperty for {L}ocally {C}ovariant {Q}uantum
	{F}ield {T}heories in {C}urved {S}pacetime.
	\newblock \href{http://dx.doi.org/10.1007/s11005-015-0798-2}{Lett. Math. Phys.
		\textbf{105}, 1633--1661 (2015)}, arXiv:1501.02682
	
	\bibitem{Fewster:MG2015}
	Fewster, C.J.: Locally covariant quantum field theory and the spin-statistics
	connection.
	\newblock \href{http://dx.doi.org/10.1142/S0218271816300159}{Internat. J.
		Modern Phys. D \textbf{25}, {163}{0015}, 18 (2016)}
	
	\bibitem{Few_Regensburg:2015}
	Fewster, C.J.: On the spin-statistics connection in curved spacetimes.
	\newblock In:
		F.~Finster, J.~Kleiner, J.~Tolksdorf, C.~Roeken (eds.) 
		\href{http://www.worldcat.org/search?q=isbn:978-3-319-26900-9}{Quantum Mathematical
		Physics: A Bridge between Mathematics and Physics.} \newblock
		Birkh\"auser, Basel (2016).
	\newblock ArXiv:1503.05797
	
	\bibitem{FewVer:dynloc_theory}
	Fewster, C.J., Verch, R.: Dynamical locality and covariance: What makes a
	physical theory the same in all spacetimes?
	\newblock \href{http://dx.doi.org/10.1007/s00023-012-0165-0}{{A}nnales
		H.~{P}oincar{\'e} \textbf{13}, 1613--1674 (2012)}
	
	\bibitem{FewVerch_aqftincst:2015}
	Fewster, C.J., Verch, R.: Algebraic quantum field theory in curved spacetimes.
	\newblock In:
		R.~Brunetti, C.~Dappiaggi, K.~Fredenhagen, J.~Yngvason (eds.) 
		\href{http://www.worldcat.org/search?q=isbn:978-3-319-21352-1}{Advances in
		Algebraic Quantum Field Theory.} 
		Mathematical Physics Studies, pp. 125--189. Springer International
	Publishing, Springer International Publishing (2015)
	
	\bibitem{Haag}
	Haag, R.: \href{http://dx.doi.org/10.1007/978-3-642-97306-2}{Local Quantum
		Physics: Fields, Particles, Algebras.} 
	\newblock Springer-Verlag, Berlin (1992)
	
	\bibitem{HaagLopSoh:1975}
	Haag, R., {\L}opusza{\'n}ski, J.T., Sohnius, M.: All possible generators of
	supersymmetries of the {$S$}-matrix.
	\newblock \href{http://dx.doi.org/10.1016/0550-3213(75)90279-5}{Nuclear Phys. B
		\textbf{88}, 257--274 (1975)}
	
	\bibitem{HacHanSch:2016}
	Hack, T.P., Hanisch, F., Schenkel, A.: Supergeometry in locally covariant
	quantum field theory.
	\newblock \href{http://dx.doi.org/10.1007/s00220-015-2516-4}{Comm. Math. Phys.
		\textbf{342}, 615--673 (2016)}
	
	\bibitem{Ho&Wa01}
	Hollands, S., Wald, R.M.: Local {W}ick polynomials and time ordered products of
	quantum fields in curved spacetime.
	\newblock \href{http://dx.doi.org/10.1007/s002200100540}{Commun. Math. Phys.
		\textbf{223}, 289--326 (2001)}
	
	\bibitem{Ho&Wa02}
	Hollands, S., Wald, R.M.: Existence of local covariant time ordered products of
	quantum fields in curved spacetime.
	\newblock \href{http://dx.doi.org/10.1007/s00220-002-0719-y}{Commun. Math.
		Phys. \textbf{231}, 309--345 (2002)}
	
	\bibitem{Isham_spinor:1978}
	Isham, C.J.: Spinor fields in four-dimensional space-time.
	\newblock \href{http://dx.doi.org/10.1098/rspa.1978.0219}{Proc. Roy. Soc.
		London Ser. A \textbf{364}, 591--599 (1978)}
	
	\bibitem{Jost:1966}
	Jost, R.: Eine {B}emerkung zu einem ``{L}etter'' von {L}. {O}'{R}aifeartaigh
	und einer {E}ntgegnung von {M}. {F}lato und {D}. {S}ternheimer.
	\newblock \href{http://dx.doi.org/10.5169/seals-113692}{Helv. Phys. Acta
		\textbf{39}, 369--375 (1966)}
	
	\bibitem{LechSan:2016}
	Lechner, G., Sanders, K.: Modular nuclearity: A generally covariant
	perspective.
	\newblock \href{http://dx.doi.org/10.3390/axioms5010005}{Axioms \textbf{5}, 5
		(2016)}
	
	\bibitem{Lopus:1971}
	{\L}opusza{\'n}ski, J.T.: On some properties of physical symmetries.
	\newblock \href{http://dx.doi.org/10.1063/1.1665551}{J. Mathematical Phys.
		\textbf{12}, 2401--2412 (1971)}
	
	\bibitem{Michel:1964}
	Michel, L.: Sur les extensions centrales du groupe de {L}orentz inhomog\`ene
	connexe.
	\newblock \href{http://dx.doi.org/10.1016/0029-5582(64)90336-0}{Nuclear Phys.
		\textbf{57}, 356--385 (1964)}
	
	\bibitem{Neeb:2007}
	Neeb, K.H.: Non-abelian extensions of infinite-dimensional {L}ie groups.
	\newblock \href{http://dx.doi.org/10.5802/aif.2257}{Ann. Inst. Fourier
		(Grenoble) \textbf{57}, 209--271 (2007)}
	
	\bibitem{ORaif:1965}
	O'Raifeartaigh, L.: Internal symmetry and {L}orentz invariance.
	\newblock \href{http://dx.doi.org/10.1103/PhysRevLett.14.332}{Phys. Rev. Lett.
		\textbf{14}, 332--334 (1965)}
	
	\bibitem{Parke:1980}
	Parke, S.: Absence of particle production and factorization of the {S}-matrix
	in 1 + 1 dimensional models.
	\newblock
	\href{http://dx.doi.org/http://dx.doi.org/10.1016/0550-3213(80)90196-0}{Nuclear
		Physics B \textbf{174}, 166 -- 182 (1980)}
	
	\bibitem{PelcHor:1997}
	Pelc, O., Horwitz, L.P.: Generalization of the {C}oleman-{M}andula theorem to
	higher dimension.
	\newblock \href{http://dx.doi.org/10.1063/1.531846}{J. Math. Phys. \textbf{38},
		139--172 (1997)}
	
	\bibitem{Pinamonti:2009}
	Pinamonti, N.: Conformal generally covariant quantum field theory: the scalar
	field and its {W}ick products.
	\newblock \href{http://dx.doi.org/10.1007/s00220-009-0780-x}{Comm. Math. Phys.
		\textbf{288}, 1117--1135 (2009)}
	
	\bibitem{Rejzner_book}
	Rejzner, K.: \href{http://dx.doi.org/10.1007/978-3-319-25901-7}{Perturbative
		algebraic quantum field theory: An introduction for mathematicians.}
		\newblock 
	\newblock Mathematical Physics Studies. Springer, Cham (2016)
	
	\bibitem{Ruzzi_punc:2005}
	Ruzzi, G.: Punctured {H}aag duality in locally covariant quantum field
	theories.
	\newblock \href{http://dx.doi.org/10.1007/s00220-005-1310-0}{Comm. Math. Phys.
		\textbf{256}, 621--634 (2005)}
	
	\bibitem{Sanders_ReehSchlieder}
	Sanders, K.: On the {R}eeh-{S}chlieder property in curved spacetime.
	\newblock \href{http://dx.doi.org/10.1007/s00220-009-0734-3}{Commun. Math.
		Phys. \textbf{288}, 271--285 (2009)}
	
	\bibitem{Sanders_dirac:2010}
	Sanders, K.: The locally covariant {D}irac field.
	\newblock \href{http://dx.doi.org/10.1142/S0129055X10003990}{Rev. Math. Phys.
		\textbf{22}, 381--430 (2010)}
	
	\bibitem{Shtern:2008}
	Shtern, A.I.: A version of van der {W}aerden's theorem and a proof of
	{M}ishchenko's conjecture for homomorphisms of locally compact groups.
	\newblock \href{http://dx.doi.org/10.1070/IM2008v072n01ABEH002397}{Izv. Ross.
		Akad. Nauk Ser. Mat. \textbf{72}, 183--224 (2008)}
	
	\bibitem{Verch01}
	Verch, R.: A spin-statistics theorem for quantum fields on curved spacetime
	manifolds in a generally covariant framework.
	\newblock \href{http://dx.doi.org/10.1007/s002200100526}{Commun. Math. Phys.
		\textbf{223}, 261--288 (2001)}
	
\end{thebibliography}
}

\end{document}